\documentclass[11pt]{article}
\usepackage{amssymb,amsmath}
\usepackage[final]{graphicx}
\usepackage{a4wide}
\usepackage{times}
\usepackage{fullpage}
\usepackage{color}
\usepackage[multiple]{footmisc}

\usepackage[linesnumbered,boxed,vlined,ruled]{algorithm2e}

\def\eps{\varepsilon}

\def\union{\cup}

\def\OPT{\mathrm{OPT}}

\newenvironment{proof}{\noindent   {\bf Proof.}}{\hspace*{\fill}$\Box$\par\vspace{2mm}}

\newtheorem{lemma}{Lemma}
\newtheorem{theorem}{Theorem}
\newtheorem{corollary}{Corollary}

\newtheorem{proposition}{Proposition}

\newtheorem{definition}{Definition}

\newcommand{\universe}{\mathsf{U}}

\newcommand{\bidder}{\beta}

\newcommand{\bi}{\mathbf{b}_{-i}}
\renewcommand{\b}{\mathbf{b}}
\newcommand{\Set}{\sigma}

\newcommand{\Ex}[1]{\mbox{\rm\bf E}\left[#1\right]}
\title{Combinatorial Auctions without Money}
\author{Dimitris Fotakis
\thanks{National Technical University of Athens, Greece, {\tt fotakis@cs.ntua.gr}.}
\and Piotr Krysta
\thanks{University of Liverpool, UK, {\tt pkrysta@liverpool.ac.uk}. This author is supported by EPSRC grant EP/K01000X/1.}
\and Carmine Ventre
\thanks{Teesside University, UK, {\tt c.ventre@tees.ac.uk}.}
}

\begin{document}

\date{}

\maketitle

\begin{abstract}
Algorithmic Mechanism Design attempts to marry computation and incentives, mainly by leveraging monetary transfers between designer and selfish agents involved. This is principally because in absence of money, very little can be done to enforce truthfulness. However, in certain applications, money is unavailable, morally unacceptable or might simply be at odds with the objective of the mechanism. For example, in Combinatorial Auctions (CAs), the paradigmatic problem of the area, we aim at solutions of maximum social welfare but still charge the society to ensure truthfulness. Additionally, truthfulness of CAs is poorly understood already in the case in which bidders happen to be interested in only two different sets of goods.

We focus on the design of incentive-compatible CAs without money in the general setting of $k$-minded bidders. We trade monetary transfers with the observation that the mechanism can detect certain lies of the bidders: i.e., we study truthful CAs with verification and without money. We prove a characterization of truthful mechanisms, which makes an interesting parallel with the well-understood case of CAs with money for single-minded bidders. We then give a host of upper bounds on the approximation ratio obtained by either deterministic or randomized truthful mechanisms when the sets and valuations are private knowledge of the bidders. (Most of these mechanisms run in polynomial time and return solutions with (nearly) best possible approximation guarantees.) We complement these positive results with a number of lower bounds (some of which are essentially tight) that hold in the easier case of public sets. We thus provide an almost complete picture of truthfully approximating CAs in this general setting with multi-dimensional bidders.
\end{abstract}

\section{Introduction}
Algorithmic Mechanism Design has as main scope the realignment of the objective of the designer with the selfish interests of the agents involved in the computation. Since the Internet, as the principal computing platform nowadays, is perhaps the main motivation to study problems in which these objectives are different, one would expect truthful mechanisms to have concrete and widespread practical applications. However, one of the principal obstacles to this is the assumption that the mechanisms use monetary transfers. On one hand, money may provoke (unreasonably) large payments \cite{ESS04}; on the other hand, while money might be reasonable in some applications, such as sponsored search auctions, little justification can be found for either the presence of a digital currency or the use of money at all. There are contexts in which money is morally unacceptable (such as, to support certain political decisions) or even illegal (as for example, in organ donations). Additionally, there are applications in which the objective of the computation collides with the presence of money.

Consider Combinatorial Auctions (CAs, for short), the paradigmatic problem in Algorithmic Mechanism Design. In a combinatorial auction we have a set $\universe$ of $m$ goods and $n$ bidders. Each bidder $i$ has a \emph{private} valuation function $v_i$ that maps subsets of goods to nonnegative real numbers ($v_i(\emptyset)$ is normalized to be $0$). Agents' valuations are monotone, i.e., for $S \supseteq T$ we have $v_i(S) \geq v_i(T)$. The goal is to find a partition $S_1,\ldots, S_n$ of $\universe$ such that $\sum_{i=1}^n v_i(S_i)$ -- the \emph{social welfare} -- is maximized. For this problem, we are in a paradoxical situation: whilst, on one hand, we pursuit the noble goal of maximizing the happiness of the society (i.e., the bidders), on the other, we consider it acceptable to charge the society itself (and then ``reduce'' its total happiness) to ensure truthfulness. CAs without money would avoid this paradox, automatically guarantee budget-balanceness (property which cannot, in general, be achieved together with social welfare maximization), and deal with budgeted bidders (a case which is generally hard to handle in presence of money).

In this paper, we focus on $k$-minded bidders, i.e., bidders are interested in obtaining one out of a collection of $k$ subsets of $\universe$. In this general setting, we want to study the feasibility of designing truthful CAs without money, returning (ideally, in polynomial time) reasonable approximations of the optimal social welfare. This is, however, an impossible task in general: it is indeed pretty easy to show that there is no better than $n$-approximate mechanisms without money, even in the case of single-item auctions and truthful-in-expectation mechanisms \cite{DG10}. We therefore focus on the model of CAs with verification, introduced in \cite{KV10}.
In this model, which is motivated by a number of real-life applications and has also been considered by economists \cite{Cel06}, bidders do not overbid their valuations on the set that they are awarded. The hope is that money can be traded with the verification assumption so to be able to design ``good'' (possibly, polynomial-time) mechanisms, which are truthful without money in a well-motivated -- still challenging -- model.


\subsection{Our contribution}
The model of CAs with verification is perhaps best illustrated by means of the following motivating scenario, discussed first in \cite{KV10}. Consider a government (auctioneer) auctioning business licenses for a set $\universe$ of cities under its administration. A business company (bidder) wants to get a license for some subset of cities (subset of $\universe$) to sell her product stock to the market. Consider the bidder's profit for a subset of cities $S$ to be equal to a unitary \emph{publicly known} product price (e.g., for some products, such as drugs, the government could fix a social price) times the number of product items available in the stocks that the company possesses in the cities comprising $S$.\footnote{Note that bidders will sell products already in stock (i.e., no production costs are involved as they have been sustained before the auction is run). This is conceivable when a government runs an auction for urgent needs (e.g., salt provision for icy roads or vaccines for pandemic diseases).} In this scenario, the bidder is strategic on her stock availability. As noted in literature in Economics \cite{Cel06}, a simple inspection on the stock consistency implies that bidders cannot overbid their profits: the concealment of existing product items in stock is costless but disclosure of unavailable ones is prohibitively costly. The assumption is 
 verification \emph{a posteriori}\footnote{A stronger model of verification would require bidders to be unable to overbid at all and not just on the awarded set. However, there appears to be weaker motivations for this model: the investment required on inspections would be considerable and rather  unrealistic.}: the inspection is carried on for the solutions actually implemented and then each bidder cannot overstate her valuation for the set she gets allocated, if any. It is important to notice that bidders can misreport sets and valuations for unassigned sets in an \emph{unrestricted} way. A formal definition of the model of CAs with verification and without money can be found in Section \ref{sec:model}.

In this model, we firstly give a complete characterization of algorithms that are incentive-compatible in both the cases in which the collections of $k$ sets, each bidder is interested in, are public (also referred to, as known bidders) and private (also known as, unknown bidders); valuations are always assumed to be private. We prove that truthfulness is characterized in this context in terms of \emph{$k$-monotone algorithms}: in the case of known bidders, if a bidder is awarded a set $S$ and augments her declaration for $S$ then a $k$-monotone algorithm must, in this new instance, grant her a set in her collection which is not worse than $S$ (i.e., a set with a valuation not smaller than her valuation for $S$). (This generalizes neatly to the case of unknown bidders.) There are two important facts we wish to emphasize about our characterizations. First, their significance stems from the fact that the corresponding problem of characterizing truthfulness for CAs with money and $k$-minded bidders is poorly understood: this is a long-standing open problem already for $k=2$, see, e.g., \cite[Chapter 12]{book}. Second, it is pretty easy to see that these notions generalize the properties of monotonicity shown to characterize truthfulness with money for single-minded bidders in \cite{MuaNis02,LOS} for known and unknown bidders, respectively. More generally, these properties of monotonicity are also proved to be sufficient to get truthful mechanisms for so-called generalized single-minded bidders \cite{BKV05}. This is an interesting development as, to the best of our knowledge, it is the first case in which a truthful mechanism with money can be ``translated'' into a truthful mechanism without money. 
The price to pay is ``only'' to perform verification to prevent certain lies of the bidders, while algorithms (and then their approximation guarantees) remain unchanged. Thus, 
in light of our results, previously known algorithms presented in, e.g., \cite{
LOS,BKV05,SODA10} 
assume a double relevance: they are truthful not only when money can be used, but also in absence of money when verification can be implemented. This equivalence gives also a strong motivation for our model. Naturally, the picture for the multi-dimensional case of $k>1$ is more blurry since, as we mention above, truthfulness with money is not well understood yet in these cases.

Armed with the characterization of truthfulness, we provide a number of upper and lower bounds on the approximation guarantee to the optimal social welfare of truthful CAs without money and with verification. The upper bounds hold for the harder case of unknown bidders. We give an upper bound of $O(b \sqrt[b]{m})$ in the case in which each good in $\universe$ has a supply $b$. This algorithm is deterministic, runs in polynomial time and adapts an idea of multiplicative update of good prices by \cite{KV12}. Following similar ideas, we also obtain randomized universally truthful mechanisms with approximation ratios of  $O(d^{1/b} \cdot \log(bm))$ and
 $O(m^{1/(b+1)} \cdot \log(bm))$, where $d$ is the maximum size of sets in the bidders' collections.
Our most significant deterministic polynomial-time upper bound is obtained, in the case of $b=1$, by
a simple greedy mechanism that exploits the characteristics of the model without money. 
This algorithm returns a $\min\{m, d+1\}$-approximate solution. These upper bounds are complemented by two simple randomized universally truthful CAs without money: the first achieves a $k$-approximation in exponential time; the second runs instead in polynomial-time and has a $O(\sqrt{m})$-approximation guarantee. We note here that all our polynomial-time upper bounds are computationally (nearly) best possible even when the algorithm has full knowledge of the bidders' data. We also would like to note that all, but the $k$-approximate, upper bounds given can be obtained in the setting in which bidders' declare so-called demand oracles, see, e.g., \cite[Chapter~11]{book}. 
%
We complete this study by proving a host of lower bounds on the approximation guarantee of truthful CAs without money for known bidders, without any computational assumption. (Note that the class of incentive-compatible algorithms
for known bidders is larger than the class for unknown bidders.) We prove the following lower bounds: $2$ for deterministic mechanisms; $5/4$ for universally truthful mechanisms; and, finally, $1.09$ for truthful-in-expectation mechanisms. This implies that the optimal mechanisms are not truthful in our model. Additionally, stronger lower bounds are proved for deterministic truthful mechanisms that use priority algorithms \cite{BL10}. These algorithms process (and take decisions) one \emph{elementary item} at the time, from a list of ordered items. The ordering can also change adaptively after each item is considered. (Note that our greedy mechanism falls in the category of non-adaptive priority algorithms since it process bids as items, which are ordered at the beginning.) We give a lower bound of $d$ for priority algorithms that process bids as elementary items (thus, essentially matching the upper bound of the greedy algorithm) and a lower bound of $m/2$ in the case in which the algorithm processes bidders as items.

Our bounds give a rather surprising picture of the relative power of verification versus money, thus suggesting that the two models are somehow incomparable. For example, we have a $O(\sqrt{m})$-approximate universally truthful mechanism, which matches the guarantee of the universally truthful mechanism with money given by \cite{Dob07}. (However, it is worth mentioning that the latter mechanism does not guarantee the approximation ratio since there is an error probability of $O(\log m/\sqrt{m})$ which cannot be reduced by, e.g., repeating the auction or otherwise truthfulness would be lost.) On the other hand, because of our lower bounds, we know that it is not possible to implement the optimal outcome without money; while, if we have exponential computational time, we can truthfully implement the optimal solution using VCG payments. However, if we restrict to polynomial-time mechanisms, then we have a deterministic truthful $\min\{m,d+1\}$-approximation mechanism without money, based on the aforementioned greedy algorithm; with money, instead, it is not known how to obtain any polynomial-time deterministic truthful mechanism with an approximation ratio better than the $O(m/\sqrt{\log{m}})$-approximation given in \cite{HKMT04}.
%
%
Moreover, \cite[Theorem~2]{BL10} proved a lower bound of $\Omega(m)$ on the approximation ratio of any truthful greedy mechanism with money for instances with demanded sets of cardinality at most $2$. Our greedy mechanism achieves an approximation ratio of $3$ for such instances, which implies that this 
lower bound 
does not hold in our model without money.
%
Additionally, we show that the greedy mechanism cannot be made truthful with money, which suggests that the model without money couples better with greedy selection rules. A general lower bound in terms of $m$ for CAs without money would shed further light on this debate of power of verification versus power of money. In this regard, we offer an interesting conjecture in Section \ref{sec:discussion}.


\subsection{Related work}
CAs as an optimization problem (without strategic consideration) is known to be NP-hard to solve optimally or even to approximate: neither an approximation ratio of $m^{1/2-\epsilon}$, for any constant $\epsilon>0$, nor of $O(d/\log d)$ can be obtained in polynomial time \cite{Nis02,LOS,HazanSS06}. As a consequence, a large body of literature has focused on the design of polynomial-time truthful CAs that return as good an approximate solution as possible, under assumptions (i.e., restrictions) on bidders' valuation domains. For single-minded domains, a host of truthful CAs have been designed (see, e.g., \cite{LOS,
MuaNis02,%
BKV05}). A more complete picture of what is known for truthful CAs under different restrictions of bidders' domains can be found in Figure 11.2 of \cite{book}.

The authors of \cite{KV10}, instead of restricting the domains of the bidders, proposed to restrict the way bidders lie. We are adopting here their model, adapting it to the case without money. The definition of CAs with verification is inspired by the literature on mechanisms with verification (see, e.g.,
\cite{NisRon99,
esa08,PenVen09} and references therein). 
Mechanism design problems where players have restrictions on the way of lying are also considered in theoretical economics. We next discuss some of the work more relevant to this paper. Green and Laffont \cite{GreenLaffont86} define and motivate a model of partial verification wherein bidders can only report bids from a type-dependent set of allowed messages; they characterize bidding domains for which the Revelation Principle holds in presence of this notion of restricted bidding. This model has been further studied by Singh and Wittman \cite{SinWit01} and later extended in \cite{CaraEC12} to allow probabilistic verification of bids outside the set of allowed messages.  The economic model that is closest to ours is the one studied in \cite{Cel06}; therein verification is supposed to take place for every outcome and not just for the implemented solution and is therefore stronger and less realistic than ours. Another related line of work tries to establish when a subset of incentive-compatibility constraints is sufficient to obtain full incentive-compatibility. \cite{Moore1984} considers a single good, single buyer optimal auction design and studies conditions under which no-overbidding constraints would also imply the full incentive compatibility of the underlying auction. Other papers studying this kind of questions are \cite{SherVohra2010,Carroll2011}. In particular, the results in \cite{Carroll2011} (and to some extent in \cite{CaraEC12}) seem to suggest that one has to focus only on ``one-sided'' verification, for otherwise a mechanism is truthful if and only if it satisfies a subset of incentive-compatibility constraints.

%
Our work fits in the framework of approximate mechanism design without money, initiated by \cite{PT09}. The idea is that for optimization problems where the optimal solution cannot be truthfully implemented without money, one may resort to the notion of approximation, and seek for the best approximation ratio achievable by 
truthful algorithms. Approximate mechanisms without money have been obtained for various problems, among them, for locating one or two facilities in metric spaces (see e.g., \cite{PT09,LSWZ10}). Due to the apparent difficulty of truthfully locating three or more facilities with a reasonable approximation guarantee,
notions conceptually similar to our notion of verification have been proposed \cite{NST10,FT10}. \cite{Kouts11} considers truthful mechanisms without money, for scheduling selfish machines whose execution times can be (strongly) verified.
%
%
%
%
The authors of \cite{DG10} consider the design of mechanisms without money for, what they call, the Generalized Assignment problem: $n$ selfish jobs compete to be processed by $m$ unrelated machines; the only private data of each job is the set of machines by which it can be actually processed. This problem can be modeled via maximum weight bipartite matching and the latter can be cast as a special case of CAs with demanded sets of cardinality $1$; then \cite[Algorithm~1]{DG10} can be regarded as a special case of our greedy algorithm.
%
%

\section{Model and preliminaries}\label{sec:model}
In a combinatorial auction we have a set $\universe$ of $m$ goods and $n$ agents, a.k.a.~bidders. Each $k$-minded XOR-bidder $i$ has a \emph{private} valuation function $v_i$ and is interested in obtaining only one set in a \emph{private} collection ${\cal S}_i$ of subsets of $\universe$, $k$ being the size of ${\cal S}_i$. The valuation function maps subsets of goods to nonnegative real numbers ($v_i(\emptyset)$ is normalized to be $0$). Agents' valuations are monotone: for $S \supseteq T$ we have $v_i(S) \geq v_i(T)$.


The goal is to find a partition $S_1,\ldots, S_n$ of $\universe$ such that $\sum_{i=1}^n v_i(S_i)$ --the \emph{social welfare}-- is maximized.
%
%
As an example consider $\universe=\{1,2,3\}$ and the first bidder to be interested in ${\cal S}_1=\left\{\{1\},\{2\},\{1,2\}\right\}$. The valuation function of bidder $i$ for 
$S \not\in {\cal S}_i$ is
\begin{equation}\label{eq:valuation}
v_i(S) = \left\{ \begin{array}{ll} \max_{S' \in {\cal S}_i:
S \supseteq S'} \{v_i(S')\} & \mbox{if } \exists S' \in {\cal S}_i \wedge S \supseteq
S', \\ 0 & \mbox{otherwise.}\end{array}\right.
\end{equation}
Accordingly, we say that $v_i(S) \neq 0$ (for $S \not\in {\cal S}_i$) is \emph{defined} by an inclusion-maximal set $S' \in {\cal S}_i$ such that $S' \subseteq S$ and $v_i(S')=v_i(S)$. If $v_i(S) =0$ then we say that $\emptyset$ defines it. So in the example above $v_1(\{1,2,3\})$ is defined by $\{1,2\}$.

Throughout the paper we assume that bidders are interested in sets of cardinality at most $d \in \mathbb{N}$, i.e.,
$d = \max \{|S|\; : \; \exists\; i\; s.t.\; S \in {\cal S}_i \wedge v_i(S) > 0 \}$.

Assume that the sets $S \in {\cal S}_i$ and the values $v_i(S)$ are private knowledge of the bidders. Then,
we want to design an \emph{allocation algorithm} ({\em auction}) that for a given input of bids from the bidders, outputs a feasible assignment (i.e., at most one of the requested sets is allocated to each bidder, and allocated sets are pair-wise disjoint). The auction should guarantee that no bidder has an incentive to misreport her preferences and maximize the social welfare (i.e., the sum of the valuations of the winning bidders).

More formally, we let ${\cal T}_i$ be a set of $k$ 
non-empty subsets of $\universe$ and let $z_i$ be the corresponding valuation function of agent $i$, i.e., $z_i : {\cal T}_i \rightarrow \mathbb{R}^+$. We call $b_i=(z_i,{\cal T}_i)$ a \emph{declaration} (or {\em bid}) of bidder $i$. We let $t_i = (v_i, {\cal S}_i)$ be the \emph{true type} of agent $i$.
We also let $D_i$ denote the set of all the possible declarations of agent $i$ and call $D_i$ the \emph{declaration domain} of bidder $i$.
Fix the declarations $\bi$ of all the agents but $i$. For any declaration $b_i = (z_i, {\cal T}_i)$ in $D_i$, we let $A_i(b_i, \bi)$ be the set that an auction $A$ on input $\b=(b_i,\bi)$ allocates to bidder $i$. If no set is allocated to $i$ then we naturally set $A_{i}(b_i, \bi)=\emptyset$. Observe that, according to \eqref{eq:valuation}, $v_i(\emptyset)=0$.
We say that $A$ is a truthful auction without money if the following holds for any $i$, $b_i \in D_i$ and $\bi$:
\begin{equation}\label{eq:truthful}
v_i(A_{i}(t_i, \bi))  \geq  v_i(A_{i}(\b)).
\end{equation}
We also define notions of truthfulness in the case of randomization: we either have universally truthful CAs, in which case the mechanism is a probability distribution over deterministic truthful mechanisms, or truthful-in-expectation CAs, where in \eqref{eq:truthful} we use the expected values, over the random coin tosses of the algorithm, of the valuations. We also say that a mechanism $A$ is an $\alpha$-approximation for CAs with $k$-minded bidders if for all $\mathbf{t}=(v_i, {\cal S}_i)_{i=1}^n$, $
\sum_{i=1}^n v_i(A_i(\mathbf{t})) \geq \OPT/\alpha,
$
$\OPT$ being the value of a solution with maximum social welfare for the instance $\mathbf{t}$.

Recall that $A_{i}(t_i,\bi)$ may not belong to the set of demanded sets ${\cal S}_i$. In particular, there can be several sets in ${\cal S}_i$ (or none) that are subsets of $A_{i}(t_i,\bi)$.
However, as observed above (cf. \eqref{eq:valuation}), the valuation is defined by a set in ${\cal S}_i \cup \{\emptyset\}$ which is an inclusion-maximal subset of set $A_{i}(t_i,\bi)$ that maximizes the valuation of agent $i$. We denote such a set as $\Set(A_i(t_i,\bi)|t_i)$, i.e., $v_i(A_{i}(t_i,\bi)) = v_i(\Set(A_i(t_i,\bi)|t_i))$. In our running example above, it can be for some algorithm $A$ and some $\mathbf{b}_{-1}$, that $A_1(t_1, \mathbf{b}_{-1})=\{1,2,3\} \not\in {\cal S}_1$ whose valuation is defined as observed above by $\{1,2\}$; the set $\{1,2\}$ is denoted as $\Set(A_1(t_1,\mathbf{b}_{-1})|t_1)$. (Similarly, we define $\Set(A_i(b_i,\bi)|b_i))\in {\cal T}_i \cup \{\emptyset\}$ w.r.t. $A_{i}(b_i,\bi)$ and declaration $b_i$.) Following the same reasoning, we let $\Set(A_i(b_i,\bi)|t_i)$ denote the set in ${\cal S}_i \cup \{\emptyset\}$ such that $v_i(A_{i}(b_i,\bi)) = v_i(\Set(A_i(b_i,\bi)|t_i))$.

We focus on \emph{exact} algorithms\footnote{
An
algorithm is exact if, to each bidder, either only one of the
declared sets is awarded or none.} in the sense of \cite{LOS}.
This means that $A_{i}(b_i,\bi) \in {\cal T}_i \cup \{\emptyset\}$.
This implies, by monotonicity of the valuations, that $A_{i}(b_i,\bi) = \Set(A_i(b_i,\bi)|b_i)$ and then
the definition of $\Set(\cdot|\cdot)$ yields the following for any 
$t_i, b_i \in D_i$:
\begin{equation}\label{eq:sets:exactness}
\Set(A_i(b_i,\bi)|t_i) \subseteq A_{i}(b_i,\bi)=\Set(A_i(b_i,\bi)|b_i).
\end{equation}

In the verification model each bidder can only declare lower
valuations for the set she is awarded. More formally, bidder
$i$ whose type is $t_i=(v_i, {\cal S}_i)$ can declare a type
$b_i=(z_i, {\cal T}_i)$ if and only if whenever $\Set(A_i(b_i,\bi)|b_i) \neq \emptyset$:
\begin{equation}\label{eq:ver:vals}
z_i(\Set(A_i(b_i,\bi)|b_i)) \leq v_i(\Set(A_i(b_i,\bi)|t_i)).
\end{equation}
In particular, bidder $i$ evaluates the assigned set $\Set(A_i(b_i,\bi)|b_i) \in {\cal T}_i$ as $\Set(A_i(b_i,\bi)|t_i) \in {\cal S}_i \cup \{\emptyset\}$, i.e., $v_i(\Set(A_i(b_i,\bi)|t_i))=v_i(\Set(A_i(b_i,\bi)|b_i))$. Thus the set $\Set(A_i(b_i,\bi)|b_i)$ can be used to verify \emph{a posteriori} that bidder $i$ has overbid
declaring $z_i(\Set(A_i(b_i,\bi)|b_i))>v_i(\Set(A_i(b_i,\bi)|$ $b_i))=v_i(\Set(A_i(b_i,\bi)|t_i))$. To be more concrete, consider the motivating scenario for CAs with verification above. The set of cities $\Set_{i}(A(b_i,\bi)|b_i)$ for which the government assigns licenses to bidder $i$ when declaring $b_i$, can be used a posteriori to verify overbidding by simply counting the product items available in the stock of the cities for which licenses were granted to bidder $i$.

When \eqref{eq:ver:vals} is not satisfied then the bidder is
caught lying by the verification step. We assume that this behavior is very undesirable for the bidder (e.g., for simplicity we can assume that in such a case the bidder loses prestige and the possibility to participate in the future auctions). This way \eqref{eq:truthful} is satisfied directly when \eqref{eq:ver:vals} does not hold (as in such a case a lying bidder would have an infinitely bad utility because of the assumption above). Thus in our model, truthfulness with verification and without money of an auction is fully captured by \eqref{eq:truthful} holding only for any $i$, $\bi$ and $b_i=(z_i, {\cal T}_i) \in D_i$ such that \eqref{eq:ver:vals} is fulfilled. Since our main focus is on this class of truthful mechanisms with verification and no money, we sometimes avoid to mention that and simply refer to truthful mechanisms/algorithms.

\paragraph{A graph-theoretic approach}
The technique we will use to derive truthful auctions
for multi-minded XOR bidders is a straightforward variation of the so-called cycle monotonicity technique. Consider an algorithm $A$. We will set up a
weighted graph for each bidder $i$ depending on $A$, bidder
domain $D_i$ and the declaration $\bi$ of all the bidders but
$i$ in which the non-existence of negative-weight edges
guarantees the truthfulness of the algorithm. This is a well known technique. More formally, fix algorithm $A$, bidder $i$ and
declarations $\bi$. The \emph{declaration graph} associated to
algorithm $A$ has a vertex for each possible declaration in the
domain $D_i$. We add an arc between $a=(z,{\cal T})$ and
$b=(w,{\cal U})$ in $D_i$ whenever a bidder of type $a$ can
declare to be of type $b$ obeying \eqref{eq:ver:vals}.
Following the definition of the
verification setting, edge $(a,b)$ belongs to the graph if and only if $z(\Set(b|a)) \geq w(\Set(b|b))$.~\footnote{To ease our notation
we let $\Set(b|a)$ be a
shorthand for $\Set(A_i(b,\bi)|a)$ when the algorithm, the bidder $i$ and declarations $\bi$ are clear from the context as in this case.}%
\footnote{Strictly speaking for an edge $(a,b)$ in the graph, we should require that $z(\Set(b|a)) \geq w(\Set(b|b))$ only whenever $\Set(b|b) \neq \emptyset$ as this set would be needed to verify. However, because of the monotonicity and normalization of valuations, $z(\Set(b|a)) \geq w(\Set(b|b))$ holds also whenever $\Set(b|b) = \emptyset$, since $\sigma(b|a)=\emptyset$ and $z(\emptyset)=w(\emptyset)=0$.} The weight of the edge $(a,b)$ is defined as
$z(\Set(a|a))-z(\Set(b|a))$ and thus encodes the loss that a bidder whose
type is $(z,{\cal T})$ incurs by declaring $(w,{\cal U})$. The
following result (whose proof is straightforward) relates the weight of edges of the declaration graph to the truthfulness of the algorithm.

\begin{proposition}\label{prop:cycles}
$A$ is a truthful auction with verification without money for CAs with $k$-minded bidders if and only if each declaration graph associated to algorithm $A$ does not have negative-weight edges.
\end{proposition}

In the case of mechanisms without verification, the graph above is complete. Such a graph can be used to check whether algorithms can be augmented with payments so to ensure truthfulness, both in the scenario with verification and without. Incentive-compatibility of algorithms is known to coincide with the case in which each graph has not negative-weight cycles \cite{Vohra07}. We will use this fact to show that certain algorithms cannot be made truthful with money.

\paragraph{Known vs Unknown $k$-minded bidders} 
In the discussion above, we consider the case in which the collection of $k$ sets, each bidder is interested in, is private knowledge. In this case, we refer to the problem of designing truthful auctions that maximize the social welfare as CAs with \emph{unknown $k$-minded bidders} (or, simply, unknown bidders). An easier scenario is the setting in which the sets are public knowledge and bidders are only strategic about their valuations. In this case, we instead talk about CAs with \emph{known $k$-minded bidders} (or, simply, known bidders). Our upper bounds hold for the more general case of unknown bidders, while the lower bounds apply to the larger class of mechanisms truthful for known bidders.

\section{Characterization of truthful mechanisms}
In this section we characterize the algorithms that are truthful in our setting, in both the scenarios of known and unknown bidders. Interestingly, the characterizing  property is algorithmic only and turns out to be a generalization of the properties used for the design of truthful CAs with money and no verification for single-minded bidders.

\subsection{Characterization for known bidders}
In this case, for each $k$-minded bidder $i$ we know ${\cal S}_i$. The following property generalizes monotonicity of \cite{MuaNis02} and characterizes truthful auctions without money and with verification.

\begin{definition}\label{def:kmono}
An algorithm $A$ is \emph{$k$-monotone} if the following holds for any $i$, any $\bi$, any $a \in D_i$: if $A_i(a,\bi)=S$ then for all $b \in D_i$ such that $b(S) \geq a(S)$ it holds $b(A_{i}(b,\bi)) \geq b(S)$.
\end{definition}

\begin{theorem}\label{thm:char:known}
An algorithm $A$ is truthful without money and with verification for known $k$-minded bidders if and only if $A$ is $k$-monotone.
\end{theorem}
\begin{proof}
Fix $i$, $\bi$ and consider the declaration graph associated to algorithm $A$. Take any edge of the graph $(b,a)$ and let $S$ denote $A_i(a,\bi)$. By definition, the edge exists if and only if $b(S) \geq a(S)$.

Now if the algorithm is $k$-monotone then we also have that $b(A_{i}(b,\bi)) \geq b(S)$ and then the weight $b(A_{i}(b,\bi)) - b(S)$ of edge $(b,a)$ is non-negative. Vice versa, assume that the weight of $(b,a)$ is non-negative: this means that whenever $b(S) \geq a(S)$ then it must be $b(A_{i}(b,\bi)) \geq b(S)$ and therefore $A$ is $k$-monotone. The theorem follows from Proposition \ref{prop:cycles}.
\end{proof}

Similarly to \cite{MuaNis02}, $k$-monotonicity implies the existence of thresholds (critical values/prices). Towards this end, it is important to consider the sets in ${\cal S}_i$ in decreasing order of (true) valuations. Accordingly, we denote ${\cal S}_i=\{S_i^{1}, \ldots, S_i^{k}\}$, with $v_i(S_i^{j}) > v_i(S_i^{l})$ if and only if $j < l$. 

\begin{lemma}\label{le:ths}
An algorithm $A$ is $k$-monotone if and only if for any $i$, any $\bi$, any $t_i$ there exist $k$ threshold values $\Theta_i^1(\bi), \ldots, \Theta_i^k(\bi)$ such that: if $b_i(S_i^j) > \Theta_i^j(\bi)$ and $b_i(S_i^\ell) < \Theta_i^\ell(\bi)$, for all $\ell < j$ then $\Set(A_i(b_i,\bi)|t_i)=S_i^j$. Moreover, if $b_i(S_i^\ell) < \Theta_i^\ell(\bi)$, for all $\ell \in [k]$ then $\Set(A_i(b_i,\bi)|t_i)=\emptyset$.
\end{lemma}
%
The lemma above assumes that bidders have $k$ different valuations for each of their minds. This is a rather nonrestrictive way to model CAs for $k$-minded bidders. In the more general case in which bidders are allowed to have ties in their valuations, one can prove that monotonicity implies the existence of thresholds, while the other direction is not true in general but only under some assumption on $A_i(b_i, \bi)$.

\subsection{Characterization for unknown bidders}
The following property generalizes the property of monotonicity of algorithms defined by \cite{LOS} and characterizes truthful auctions without money and with verification.

\begin{definition}\label{def:ksetmono}
An algorithm $A$ is \emph{$k$-set monotone} if the following holds for any $i$, any $\bi$ and any $a=(z, {\cal T}) \in D_i$: if $A_i(a,\bi) = T$ then for all $b=(w, {\cal U})$ such that $\Set(T|b)=U$, $w(U) \geq z(T)$ we have $A_i(b,\bi)=S$ with $w(S) \geq w(U)$.
\end{definition}
To see how this notion generalizes \cite{LOS}, it is important to understand what is $U$. In detail, $\sigma(T|b)=U$, in the above definition, should be read as to indicate that bidder $i$ going from declaration $a$ to declaration $b$, substituted $T \in {\cal T}$ with $U \in {\cal U}$ and $U \subseteq T$. This is because $\sigma(T|b)$ denotes the set in the collection of sets demanded by a bidder of type $b$ which defines the valuation of $T$. Specifically, $U \in {\cal U}$ is such that $w(U)=w(T)$. (Note that if $T$ belonged to $\cal U$ then $U$ would be $T$ itself.)

\begin{theorem}\label{thm:char:unknown}
An algorithm $A$ is truthful without money and with verification for  $k$-minded bidders if and only if $A$ is $k$-set monotone.
\end{theorem}
\begin{proof}
Fix $i$, $\bi$ and consider the declaration graph associated to algorithm $A$. Take any edge of the graph $(b=(w, {\cal U}),a=(z, {\cal T}))$ and let $T$ denote $A_i(a,\bi)$. By definition, the edge exists if and only if $w(U) \geq z(T)$, with $U=\sigma(T|b)$.

Now if the algorithm is $k$-set monotone then we also have that $w(A_{i}(b,\bi)) \geq w(U)$ and then the weight $w(A_{i}(b,\bi)) - w(U)$ of edge $(b,a)$ is non-negative. Vice versa, assume that the weight of $(b,a)$ is non-negative: this means that whenever $w(U) \geq z(T)$ then it must be $w(A_{i}(b,\bi)) \geq w(U)$ and therefore $A$ is $k$-set monotone. The theorem follows from Proposition \ref{prop:cycles}.
\end{proof}

Observe, that our characterization of Theorem \ref{thm:char:unknown}
for unknown single-minded bidders implies the existence of a threshold for any set. Namely, let $A$ be a given $1$-set monotone algorithm, and let $i$ be a fixed bidder with declaration $(z, {\cal T}) \in D_i$. Then for the set $T \in {\cal T}$ (here, $|{\cal T}|$ = 1), algorithm $A$ is monotone with respect to $z(T)$ and thus there exists a critical threshold. It is not hard to see that thresholds exist also for unknown $k$-minded bidders, with $k>1$.

The result in Theorem \ref{thm:char:unknown} also relates to the characterization of truthful CAs with money and no verification 
(see, e.g., Proposition~9.27 in \cite{book}). While the two characterizations look pretty similar, there is an important difference: in the setting with money and no verification, each bidder optimizes her valuation minus the critical price over all her demanded sets; in the setting without money and with verification, each bidder optimizes only her valuation over all her demanded sets among those that are bounded from below by the threshold.

\subsection{Implications of our characterizations}
We discuss here two conceptually relevant consequences of our results above. In a nutshell, a reasonably large class of truthful mechanisms with money can be turned into truthful mechanisms without money, by using the verification paradigm.

\subsubsection{Single-minded versus multi-minded bidders}

Observe that our characterization of truthful mechanisms without money for CAs with $1$-minded bidders with known and unknown bidders is exactly the same as the characterization of truthful mechanisms with money in this setting, see, e.g., pages 274--275 in \cite{book}. This means that the two classes of truthful mechanisms in fact coincide. More formally, we have:


\begin{proposition}\label{prop:1-minded-equivalence}
Any (deterministic) truthful  $\alpha$-approximation mechanism with money for single-minded CAs can be turned into a (deterministic) truthful  $\alpha$-approximation mechanism without money with verification for the same problem, and vice versa. 
This holds 
for single-minded CAs with either known or unknown bidders.
\end{proposition}
%

\subsubsection{Beyond CAs}

It is known that a slight
generalization of monotonicity 
of \cite{
LOS} 
is a sufficient property to obtain truthful mechanisms with money also for problems involving \emph{generalized single-minded bidders} \cite{BKV05}. Intuitively, generalized single-minded bidders have $k$ private numbers associated to their type: their valuation for a solution is equal to the first of these values or minus infinity, depending on whether the solution asks the agent to ``over-perform'' on one of the other $k-1$ parameters, see \cite{BKV05} 
for details. By Theorem \ref{thm:char:unknown}, all the truthful mechanisms with money designed for this quite general type of bidders can be turned into truthful mechanisms without money, when the verification paradigm is justifiable. As a direct corollary of our characterization, we then have a host of truthful mechanisms without money and with verification for the (multi-objective optimization) problems studied in \cite{BKV05,SODA10}.

\section{Upper bounds for unknown bidders}
In this section we present our upper bounds for CAs with unknown $k$-minded bidders. 

\subsection{CAs with arbitrary supply of goods}
In this section, we consider the more general case in which elements in $\universe$ are available in $b$ copies each. Note that the characterizations above hold also in this multi-unit case. We present three polynomial-time algorithms, which are truthful for CAs with unknown bidders: the first is deterministic, the remaining are randomized and give rise to universally truthful CAs.


\begin{algorithm}[b]\label{f:MPU-algo}
\DontPrintSemicolon
 \caption{Multiplicative price update algorithm}

For each good $e \in \universe$ do $p_e^1 := p_0$. \;

For each bidder $i = 1, 2, \ldots, n$ do \;
 \hspace{20pt} Set $S_i := \mbox{argmax } \{v_i(S): S \in {\cal S}_i \mbox{ such that } v_i(S) \geq \sum_{e \in S} p^i_e\}$.\label{l:MPU:setchoice}\;
 \hspace{20pt} Update for each good $e \in S_i$: $p_e^{i+1} := p_{e}^i \cdot r$.\;
Return $S = (S_1, S_2, \ldots, S_n)$. \;
\end{algorithm}

\subsubsection{Deterministic truthful CAs}\label{s:MPU}
We adapt here the overselling multiplicative price update algorithm and its analysis from \cite{KV12} to our setting without money. The algorithm considers bidders in an arbitrary given order. We assume that the algorithm is given a parameter $\mu \ge 1$ such that $\mu/2 \leq v_{\max} < \mu$.
We will assume that such $\mu$ is known to the mechanism, and afterwards we will modify our mechanism and show how to truthfully guess $v_{\max}$.

Algorithm~\ref{f:MPU-algo} processes the bidders in an arbitrary given order, $i = 1,2,\ldots,n$. The algorithm starts with some relatively small, uniform price $p_0 = \frac{\mu}{4bm}$ of each item. When considering bidder $i$, the algorithm uses the current prices as defining thresholds and allocates to bidder $i$ a set $S_i$ in her demand set $ {\cal S}_i$ that has the maximum valuation $v_i(S_i)$ among all her sets with valuations above the thresholds. Then the prices of the elements in the set $S_i$ are increased by a factor $r$ and the next bidder is considered. 

 Let $\ell_{e}^i$ be the number of copies of good $e \in \universe$ allocated to all bidders preceding bidder $i$ and $\ell_{e}^* = \ell_{e}^{n+1}$ denote the total allocation of good $e$ to  all bidders. Let, moreover, $p_{e}^* = p_0 \cdot r^{\ell_{e}^*}$ be good $e$'s price at the end of the algorithm.

We claim now that if $p_0$ and $r$ are chosen so that $p_0 r^b = \mu$, then the allocation $S = (S_1, \ldots, S_n)$ output by Algorithm \ref{f:MPU-algo} is feasible, that is, it assigns at most $b$ copies of each good to the bidders. The argument is as follows. Consider any good $e \in \universe$. Notice that when the $b$-th copy of good $e$ is sold to any bidder then its price is updated to $p_0 r^b = \mu > v_{\max}$. Thus, good $e$ alone has a price which is above the maximum valuation of any bidder, and so no further copy will be sold.

Next we prove two 
lower bounds on the social welfare $v(S) = \sum_{i=1}^n v_i(S_i)$ of the sets $S_1,\ldots,S_n$ chosen by Algorithm~\ref{f:MPU-algo}. Let $\OPT$ denote the optimal social welfare, and recall that $p_{e}^*$ denotes the final price of good $e \in \universe$.

\begin{lemma}\label{lemma-MPU-approx1}
It holds $v(S) \geq \frac{1}{r-1} \left(\sum_{e \in U} p_{e}^* - m p_0\right)$
and $v(S) \geq \OPT - b \sum_{e \in U} p_{e}^*$.
\end{lemma}

Combining the above bounds 
yields the following result for the algorithm.

\begin{theorem}\label{theorem:improved-Bartal}
Algorithm \ref{f:MPU-algo} with $p_0 = \frac{\mu}{4 bm}$ and  $r = (4bm)^{1/b}$ produces a feasible allocation $S$ such that
$v(S) \ge  \frac{\OPT}{2(b(r-1) + 1)} \ge \frac{\OPT}{O(b \cdot (m)^{1/b})}$.
\end{theorem}

\begin{proof}
Feasibility follows from the fact that $p_0 r^b = \mu$. The first bound of Lemma~\ref{lemma-MPU-approx1} gives
$b (r-1) v(S) \geq b \sum_{e \in U} p_{e}^* - b m p_0$, which by 
the second bound is
$b (r-1) v(S) \geq b \sum_{e \in U} p_{e}^* - b m p_0 \geq \OPT - v(S) - b m p_0 \geq \OPT/2 - v(S)$, where the last inequality follows by $v_{\max} \leq \OPT$.
This finally gives us $v(S) \ge \frac{\OPT}{2(b(r-1) + 1)}$. 
\end{proof}

\begin{theorem}\label{thm:UB:supply}
Algorithm \ref{f:MPU-algo} is a truthful mechanism without money and with verification for CAs with unknown $k$-minded bidders.
\end{theorem}

\begin{proof}
Fix $i$ and $\bi$. As in Definition \ref{def:ksetmono}, take two declarations of bidder $i$, $a=(z, {\cal T})$ and $b=(w, {\cal U})$ with $w(U) \geq z(T)$, where $T=A_i(a, \bi)$ and $U=\sigma(T|b)$. (In this proof, $A$ denotes Algorithm \ref{f:MPU-algo}.) Recall that $U \subseteq T$ and $U \in {\cal U}$.

Note that the ordering is independent of the bids and then when $i$ is considered the prices $p^i_e$ for the elements $e$ of $\universe$ are the same in both $A(a,\bi)$ and $A(b, \bi)$. Since $T=A_i(a, \bi)$, we note that $z(T) \geq \sum_{e \in T} p_e^i$. This yields,
$ 
w(U) \geq z(T) \geq \sum_{e \in T} p_e^i \geq \sum_{e \in U} p^i_e.
$ 
This implies that when $A(b, \bi)$ executes line \ref{l:MPU:setchoice}, the set $U$ is taken into consideration and we can therefore conclude that $w(A_i(b, \bi)) \geq w(U)$. This shows that $A$ is $k$-set monotone and then, by Theorem \ref{thm:char:unknown}, our claim.
\end{proof}

 \begin{algorithm}[tb]\label{f:MPU-algo-modified}
\DontPrintSemicolon
 \caption{Modified multiplicative price update algorithm}
For each bidder $i \in \{1,2, \ldots, n\}$, let $v_{\max}^i$ be the valuation of $i$'s most valuable set. \;

Let $j \in \{1,2, \ldots, n\}$ be the bidder with highest value $v_{\max}^j$ (smallest index in case of ties).\;

Let  $p_0 = \frac{\mu}{4bm}$, where $\mu = (1+\epsilon) v^j_{\max}$, for a fixed $0 < \epsilon \ll 1$. \;

For each good $e \in \universe$ do $p_e^j := p_0$. \;

Let for any $i = j, 1, 2, 3, 4, \ldots , j-1, j+1, \ldots, n$, $\mbox{next}(i)$ be the next number in this order, e.g.,
$\mbox{next}(j) = 1, \mbox{next}(1) = 2, \ldots, \mbox{next}(j-1) = j+1, \ldots, \mbox{next}(n-1) = n,  \mbox{next}(n) = n+1$. \;

For each bidder $i = j, 1, 2, \ldots, j-1, j+1, \ldots, n$ do \;
 \hspace{20pt} Set $S_i := \mbox{argmax } \{v_i(S): S \in {\cal S}_i \mbox{ such that } v_i(S) \geq \sum_{e \in S} p^i_e\}$.\label{l:MPU-mod:setchoice}\;
 \hspace{20pt} Update for each good $e \in S_i$: $p_e^{\mbox{next}(i)} := p_{e}^i \cdot r$.\;
Return $S = (S_1, S_2, \ldots, S_{j-1}, S_j, S_{j+1}, \ldots, S_n)$. \;
\end{algorithm}

We now modify Algorithm \ref{f:MPU-algo} in order to remove the assumption on the knowledge of $\mu$. The modified algorithm is presented as Algorithm \ref{f:MPU-algo-modified}. We have the following result.

\begin{theorem}\label{thm:UB:supply:mod}
Algorithm \ref{f:MPU-algo-modified} is a truthful mechanism without money and with verification for CAs with unknown $k$-minded bidders. Its approximation ratio is $O(b \cdot (m)^{1/b})$.
\end{theorem}

\begin{proof}
Approximation ratio and feasibility of the produced solution follow from the choice of $p_0$ and from setting $\mu=(1+\epsilon) v^j_{\max}$, for  $0<\epsilon \ll 1$. Indeed, we can use the previous analysis of Algorithm \ref{f:MPU-algo} that did not make any assumption on the order in which bidders are processed and only required $\mu/2\leq v_{\max}<\mu$. 

We will argue now about truthfulness of the modified algorithm. Let us call bidder $j$ in Algorithm \ref{f:MPU-algo-modified}, the \emph{max bidder}. We first observe that bidder $j$ is allocated the set in her (reported) demand with highest (reported) valuation. This is because her declaration of $v^j_{\max}$ for her best set, say $Q$, is larger than $\sum_{e \in Q} p^j_e = |Q| \cdot  p_0 = |Q| \cdot \frac{(1+\epsilon) v^j_{\max}}{4bm}$ since $|Q| \leq m$. Now, fix $i$ and $\bi$. As in Definition \ref{def:ksetmono}, take two declarations of bidder $i$, $a=(z, {\cal T})$ and $b=(w, {\cal U})$ with $w(U) \geq z(T)$, where $T=A_i(a, \bi)$ and $U=\sigma(T|b)$. (In this proof, $A$ denotes Algorithm \ref{f:MPU-algo-modified}.) Recall that $U \subseteq T$ and $U \in {\cal U}$. Let $j_a$ (resp., $j_b$) be the max bidder for the bid vector $(a, \bi)$ (resp., $(b, \bi)$). 
We 
distinguish three cases.

\noindent \underline{Case 1: $i = j_a$.} In this case, $z(T)$ is larger than all the valuations in $\bi$. Since $w(U) \geq z(T)$ and since $\bi$ is unchanged then $w(U)$ is also larger than all the valuations in $\bi$, which yields $i =j_b$. But then, as observed above, since $i$ is the max bidder in $(b, \bi)$  she will get her best set in ${\cal U}$ and therefore $w(A_i(b, \bi)) \geq w(U)$.

\noindent \underline{Case 2: $i = j_b$.} Since $i$ is the max bidder in $(b, \bi)$ then we can argue, as above, that she will get her best set and so we have $w(A_i(b, \bi)) \geq w(U)$.

\noindent \underline{Case 3: $i \neq j_a, j_b$.} Since the other bids are unchanged, in this case, we have $j_a=j_b$. This implies that the ordering in which bidders are considered is the same in both $A(a, \bi)$ and $A(b, \bi)$ which in turns implies that the prices $p^i_e$ considered by the algorithm in line \ref{l:MPU-mod:setchoice} are the same in both instances. We can then use the same arguments used in the proof of Theorem \ref{thm:UB:supply} to conclude that $w(A_i(b, \bi)) \geq w(U)$.

In all the three cases we have shown that the algorithm is $k$-set monotone and then the claim follows from Theorem \ref{thm:char:unknown}.
\end{proof}

\begin{algorithm}[b]\label{f:MPU-algo-rand}
\DontPrintSemicolon
 \caption{Multiplicative price update algorithm with oblivious randomized rounding.}

For each good $e \in \universe$ do $p_e^1 := p_0$, $b_e^1  := b$. \;

For each bidder $i = 1, 2, \ldots, n$ do \;
 \hspace{20pt} Let $\universe _i = \{e \in \universe \,|\, b_e^i > 0\}$. \;
 \hspace{20pt} Set $S_i := \mbox{argmax } \{v_i(S): S \in {\cal S}_i \mbox{ such that } S \subseteq \universe _i \mbox{ and } v_i(S) \geq \sum_{e \in S} p^i_e\}$.\label{l:MPU-algo-rand:setchoice}\;
 \hspace{20pt} Update for each good $e \in S_i$: $p_e^{i+1} := p_{e}^i \cdot r$.\;
 \hspace{20pt} With probability $q$ set $R_i := S_i$ else $R_i := \emptyset$.\;
 \hspace{20pt} Update for each good $e \in R_i$: $b_e^{i+1} := b_e^{i} - 1$.\;
Return $R = (R_1, R_2, \ldots, R_n)$. \;
\end{algorithm}

\subsubsection{Randomized truthful CAs}
We 
show here how to use Algorithm \ref{f:MPU-algo-modified} to obtain randomized universally truthful mechanisms 
with expected approximation ratios of $O(d^{1/b} \log(bm))$ and $O(m^{1/(b+1)} \log(bm))$, respectively.

Observe first that if we execute Algorithm \ref{f:MPU-algo} with a smaller update factor $r = 2^{1/b}$, 
then the output solution 
allocates at most $s b$ copies of each good to the bidders, where
$s = \log(4 b m)$ \cite[Lemma 1]{KV12}. This simply follows from the fact that if $sb$ copies of good $e \in \universe$ were sold, then its price is
$p_0 2^{\log (4bm)} = \mu > v_{\max}$.
But this infeasible solution is an $O(1)$-approximation to the optimal feasible solution: plugging $r=2^{1/b}$ in the approximation ratio of $2(b(r-1)+1)$ in Theorem \ref{theorem:improved-Bartal} indeed implies an $O(1)$-approximation (see also \cite[Theorem 1]{KV12}).
This idea leads to the following randomized algorithm in \cite{KV12}: use $r = 2^{1/b}$, explicitly maintain feasibility of the produced solution, and define $q = 1/(2 {\rm e} d^{1/b} \log(4 b m))$ (where ${\rm e} \approx 2.718$) as the probability of allocating the best set to a bidder. (See Algorithm \ref{f:MPU-algo-rand} for a precise description.) 
 We now introduce the same randomization idea into our Algorithm \ref{f:MPU-algo-modified}. The resulting algorithm is Algorithm \ref{f:MPU-algo-modified-rand},
where we assume $r = 2^{1/b}$.

 \begin{algorithm}[tb]\label{f:MPU-algo-modified-rand}
\DontPrintSemicolon
 \caption{Modified multiplicative price update algorithm with randomized rounding.}
For each bidder $i \in \{1,2, \ldots, n\}$, let $v_{\max}^i$ be the valuation of $i$'s most valuable set. \;

Let $j \in \{1,2, \ldots, n\}$ be the bidder with highest value $v_{\max}^j$ (smallest index in case of ties).\;

Let  $p_0 = \frac{\mu}{4bm}$, where $\mu = (1+\epsilon) v^j_{\max}$, for a fixed $0 < \epsilon \ll 1$. \;

For each good $e \in \universe$ do $p_e^j := p_0$,  $b_e^1  := b$. \;

Let for any $i = j, 1, 2, 3, 4, \ldots , j-1, j+1, \ldots, n$, $\mbox{next}(i)$ be the next number in this order, e.g.,
$\mbox{next}(j) = 1, \mbox{next}(1) = 2, \ldots, \mbox{next}(j-1) = j+1, \ldots, \mbox{next}(n-1) = n,  \mbox{next}(n) = n+1$. \;

For each bidder $i = j, 1, 2, \ldots, j-1, j+1, \ldots, n$ do \;
 \hspace{20pt} Let $\universe _i = \{e \in \universe \,|\, b_e^i > 0\}$. \;
 \hspace{20pt} Set   $S_i := \mbox{argmax } \{v_i(S): S \in {\cal S}_i \mbox{ such that } S \subseteq \universe _i \mbox{ and } v_i(S) \geq \sum_{e \in S} p^i_e\}$. \label{l:MPU-algo-mod-rand:setchoice}\;
 \hspace{20pt} Update for each good $e \in S_i$: $p_e^{\mbox{next}(i)} := p_{e}^i \cdot r$.\;
 \hspace{20pt} If $i=j$ then set $R_i := S_i$ else (with probability $q$ set $R_i := S_i$ else $R_i := \emptyset$).\;
 \hspace{20pt} Update for each good $e \in R_i$: $b_e^{\mbox{next}(i)} := b_e^{i} - 1$.\;
Return $R = (R_1, R_2,  \ldots, R_{j-1}, R_j, R_{j+1}, \ldots, R_n)$. \;
\end{algorithm}

\begin{theorem}\label{thm:UB:supply:mod:rand}
Algorithm \ref{f:MPU-algo-modified-rand} is a universally truthful mechanism without money and with verification for CAs with unknown $k$-minded bidders. Its expected approximation ratio is $O(d^{1/b} \cdot \log(bm))$.
\end{theorem}

\begin{proof}
Approximation guarantee and feasibility of the output solution $R$ follows from essentially the same arguments used in \cite{KV12}. (For completeness we give this proof in appendix.)
We will argue now about universal truthfulness of Algorithm~\ref{f:MPU-algo-modified-rand}. This algorithm can be viewed as a probability distribution over deterministic algorithms. Each such 
algorithm, call it $A$, is defined by a $0/1$-vector $a \in \{0,1\}^{n-1}$ and 
first selects and serves the max
bidder $j$ and then serves the remaining $n-1$ bidders $1, 2, \ldots, j-1, j+1, \ldots, n$. When serving bidder $i \neq j$, 
algorithm $A$ deterministically allocates set $S_i$ to bidder $i<j$ if and only if $a_i = 1$ and to bidder $i>j$ if and only if $a_{i-1} = 1$. 
Thus, algorithm $A$ 
is Algorithm \ref{f:MPU-algo-modified}, with 
$a = (1,1, \ldots, 1)$. 
So, to show that $A$ is (deterministically) truthful we use the same argument 
of the proof  of Theorem \ref{thm:UB:supply:mod} and the additional observation that 
bidders whose corresponding bit in the vector $a$ is $0$ have 
no incentive to lie, since they are not served anyway.
\end{proof}

Finally we can also obtain a universally truthful mechanism in case  demanded sets have unbounded sizes.

\begin{theorem}\label{thm:UB:supply:mod:rand:any-m}
There exist  a universally truthful mechanism without money and with verification for CAs with unknown $k$-minded bidders with an expected approximation ratio of $O(m^{1/(b+1)} \cdot \log(bm))$.
\end{theorem}


\subsection{CAs with single supply}
We now go back to the case in which the goods in $\universe$ are provided with single supply. We present three incentive-compatible CAs: the first is deterministic, the remaining two are randomized. Among these three mechanisms, only two run in polynomial time.

\subsubsection{Greedy algorithm}
We now present a simple greedy algorithm for CAs where the supply $b=1$, see Algorithm \ref{greedyalg}. 
(Note that for goods with arbitrary supply $b$, the greedy algorithm cannot do better than Algorithm \ref{f:MPU-algo-modified} because of the lower bound of $\sqrt{m}$ in \cite{Kry05}.)
Recall that each bidder $i = 1,2, \ldots, n$ declares $(v_i, {\cal S}_i)$, where $ {\cal S}_i$ is a collection of $k$ sets bidder $i$ demands and $v_i(S)$ is the valuation of set $S \in {\cal S}_i$. Observe  that sets $S_1, \ldots, S_l$ are all the sets demanded by all bidders (with non-zero bids), i.e., $\{S_1, \ldots, S_l\} = \mathcal{S}_1 \cup \ldots \cup \mathcal{S}_n$.


\begin{algorithm}[tb]\label{greedyalg}
\caption{The greedy algorithm.}

Let $l$ denote the number of different bids, $l = n k$.

Let $b_1,b_2, \dots,b_l$ be the non-zero bids and $S_1, \ldots, S_l$ be the corresponding sets, ordered such that
$b_1 \ge \ldots \ge b_l$. In case of ties between declarations of different bidders consider first the smaller index bidder.

For each $j = 1,\ldots,l$ let $\bidder(j) \in \{1,\ldots,n\}$ be the bidder bidding $b_j$ for the set $S_j$.

${\cal P} := \emptyset$, ${\cal B} := \emptyset$. \label{alg3-line-init}


For $i=1,\dots,l$ do \label{alg3-iter-i}

\hspace{10pt} If ${\bidder}(i) \not\in {\cal B} \wedge S_i \cap S = \emptyset$ for all $S$ in ${\cal P}$ then \label{alg3-feasibility}
%
(a) ${\cal P} := {\cal P} \cup \{S_i\}$, (b) ${\cal B} := {\cal B} \cup \bidder(i)$. \label{alg3-line-update}

Return ${\cal P}$.
\end{algorithm}

We will use the linear programming duality theory to prove the approximation guarantees of our algorithm.
 Let us denote the set family ${\cal S} = \cup_{i=1}^n {\cal S}_i$,
 where bidder $i$ demands sets ${\cal S}_i$. For a given set $S \in {\cal S}_i$
 we denote by $b_i(S)$ the bid of bidder $i$ for that set. 
 Let $[n]$ 
 be the set $\{1,\ldots,n\}$.
 The LP relaxation of our problem is:
{\allowdisplaybreaks
  \begin{eqnarray}
  \max        & \sum_{i = 1}^n\sum_{S \in {\cal S}_i} b_i(S) x_i(S) & \label{e:obj} \\
  \mbox{s.t.} & \sum_{i = 1}^n \sum_{S: S \in {\cal S}_i,e \in S}  x_i(S) \leq 1 & \forall e \in \universe \label{e:capacity} \\
              & \sum_{S \in {\cal S}_i} x_i(S) \leq 1 & \forall i \in [n] 
              \label{e:demand} \\
              & x_i(S) \geq 0 & \hspace{-24pt} \forall i \in [n] 
              \forall S \in {\cal S}_i, \label{e:multi}
 \end{eqnarray}
 }%
 The corresponding dual linear program is then the following:
{\allowdisplaybreaks
\begin{eqnarray}
  \min        & \sum_{e \in \universe} y_e  + \sum_{i=1}^n z_i& \label{e:dual-obj} \\
  \mbox{s.t.} & z_i + \sum_{e \in S} y_e \geq b_i(S) & \hspace{-4pt} \forall i \in [n] 
  \,\,\, \forall S \in {\cal S}_i \label{e:dual-lbd} \\
              & z_i, y_e \geq 0 & \hspace{-4pt} \forall i \in [n] 
              \,\,\, \forall e \in \universe. \label{e:dual-multi}
\end{eqnarray}
}%
In this dual linear program dual variable
$z_i$ corresponds to the constraint (\ref{e:demand}).

\begin{theorem}\label{t:greedy-approx}
 Algorithm \ref{greedyalg} is a $\min\{m,d+1\}$-approximation algorithm for CAs with $k$-minded bidders.
\end{theorem}

\begin{proof}
   Suppose that Algorithm \ref{greedyalg}  has terminated and output solution ${\cal P}$.
   Let $SAT_{\cal P} = \cup_{S \in {\cal P}} S$. Notice that for each set $S \in {\cal S}$ that was not chosen to the
   final solution ${\cal P}$, there either is an element $e \in SAT_{\cal P} \cap S$
   which was the \emph{witness} of that event during the execution of the
   algorithm, or there exists a bidder $i$ and set $S' \in {\cal P}$ such that $S', S \in {\cal S}_i$.
         For each set $S \in {\cal S} \setminus {\cal P}$ we keep in
   $SAT_{\cal P}$ one witness for $S$. In case if there is more than one witness in $SAT_{\cal P} \cap S$,
  we keep in $SAT_{\cal P}$ the (arbitrary) witness for $S$ that belongs to the set among sets $\{T \in {\cal P}: SAT_{\cal P} \cap S \cap T \not = \emptyset\}$
  that was considered first by the greedy order. We discard the remaining elements from $SAT_{\cal P}$.

   Let us also denote ${\cal P}(S) = S \cap SAT_{\cal P}$ if $S \cap SAT_{\cal P} \not = \emptyset$ and
   ${\cal P}(S) = S$ if $S \cap SAT_{\cal P} = \emptyset$.

   Observe  first that if $m=1$, then any feasible solution just has a single set assigned to a single bidder and thus the algorithm outputs an
   optimal solution, as required. 
  
  We then assume that $m \geq 2$.  We now define a dual solution during the execution of Algorithm \ref{greedyalg}.
   We need to know the output solution ${\cal P}$ for the definition
   of this dual solution, which is needed only for analysis.
   In line \ref{alg3-line-init} of Algorithm \ref{greedyalg} we initialize these variables: $y_e := 0$ for
   all $e \in \universe$ and $z_i := 0$ for all $i \in [n]$. We add the following in line \ref{alg3-line-update}(a) of Algorithm \ref{greedyalg}:
   $
     y_e := \Delta^{S_i}_e,
     \mbox{ for all } e \in {\cal P}(S_i), \mbox{ where } \Delta^{S_i}_e = \frac{b_{\bidder(i)}(S_i)}{|{\cal P}(S_i)|}, \mbox{ for } e \in {\cal P}(S_i).
   $ Note, that for $e \in S_i \setminus SAT_{\cal P}$ the value
   of $y_e$ is not updated and remains zero.
      We also add the following instruction in line \ref{alg3-line-update}(a) of Algorithm \ref{greedyalg}:
   $
     z_{\bidder(i)} := b_{\bidder(i)}(S_i).
   $


    It is obvious that the dual solution provide a lower bound on
   the cost of the output solution:
   \begin{equation}\label{e:dual-lb12}
    \sum_{e \in \universe} y_e \leq \sum_{S_i \in {\cal P}} b_{\bidder(i)}(S_i).
   \end{equation}

     We will show now that the scaled solution $(d' \cdot y, z)$ is feasible
   for the dual linear program, where $d' = \min\{d, m-1\}$. We need to show that constraints
   (\ref{e:dual-lbd}) are fulfilled for all sets $S \in {\cal S}$. Thus, we have to prove that, for each set $S \in {\cal S} \cap {\cal S}_i$,
   \begin{equation}\label{e:approx-ratio}
      z_i + d' \sum_{e \in S} y_e \geq b_i(S).
   \end{equation}

    Suppose first that $S=S_r \in {\cal S} \setminus {\cal P}$, and let $\bidder(r) = i$. There are two possible
   reasons that set $S$ has not been included in the solution ${\cal P}$: (i) {\bf Case (a)}: there must be an element $e \in SAT_{\cal P}$ such that $e \in S$, or (ii) {\bf Case (b)}: there is another set $S' \in {\cal P}$ with $S, S' \in {\cal S}_i$.

   Let us first consider Case (a). In that case adding set $S$ to solution
   ${\cal P}$ would violate constraint (\ref{e:capacity}).
   Let $S''=S_j \in {\cal P}$ be the set in the solution that contains element $e$ and let $h = \bidder(j)$.

   Recall that $e \in S \cap S''$, thus
   $
     \sum_{e' \in S} y_{e'} \geq y_{e} = \Delta^{S''}_e = \frac{b_h(S'')}{|{\cal P}(S'')|} \geq \frac{b_h(S'')}{d}
     \geq \frac{b_i(S)}{d},
   $ where the last inequality follows from the greedy selection rule and definition of the witnesses.
   In the case if $|S| = m$, that is, $S=\universe$, we obtain that  $
     \sum_{e' \in S} y_{e'} \geq  \sum_{e'' \in S''} y_{e''}  =  \sum_{e'' \in S''} \Delta^{S''}_{e''} = b_h(S'') \geq b_i(S),
   $ where the last inequality is by  the greedy selection rule. Because $m \geq 2$, this proves (\ref{e:approx-ratio}) in Case (a).

  We consider now Case (b).
  Suppose that $S=S_r \in {\cal S} \setminus {\cal P}$ and there is another set
  $S'=S_j \in {\cal P}$ with $S, S' \in {\cal S}_i$. In this case we have $i=\bidder(j)=\bidder(r)$. Observe that when set $S'$ was chosen by Algorithm \ref{greedyalg} the dual variable $z_i$ was updated
  in line \ref{alg3-line-update}(a) as follows: $z_i = b_i(S')$. Now, because set $S'$ was considered by the
  algorithm before set $S$ we have $z_i = b_i(S') \geq b_i(S)$ by the greedy selection rule, which implies (\ref{e:approx-ratio}) in this case.

  Notice that claim (\ref{e:approx-ratio}) follows immediately from the
  definition of $z_{i}$ if set $S \in {\cal S}_{i}$ has been chosen by our algorithm, that
  is, $S \in {\cal P}$. This concludes the proof of (\ref{e:approx-ratio}).

     Finally, we put all the pieces together. We have shown that the
  dual solution $(d' \cdot y, z)$ is feasible for the dual
  linear program and so by weak duality
  $\sum_{i=1}^n z_i + d' \sum_{e \in \universe} y_e$ is an upper
  bound on the value of the optimal integral solution to our
  problem. We have also shown in (\ref{e:dual-lb12}), that
  $\sum_{e \in \universe} y_e \leq \sum_{S_i \in {\cal P}} b_{\bidder(i)}(S_i)$. Therefore, by letting $\OPT$ denote the optimal social welfare, we obtain that
{\allowdisplaybreaks
  \begin{align*}
    \OPT &\leq \sum_{i=1}^n z_i + d' \sum_{e \in \universe} y_e =
          \sum_{S_i \in {\cal P}} z_{\bidder(i)} + 
          d' \sum_{e \in \universe} y_e
    \\ & \leq \sum_{S_i \in {\cal P}} b_{\bidder(i)}(S_i) + d' \sum_{S_i \in {\cal P}} b_{\bidder(i)}(S_i) 
         =  (d'+1) \sum_{S_i \in {\cal P}} b_{\bidder(i)}(S_i).
\end{align*}}
\end{proof}

We now prove the truthfulness of Algorithm \ref{greedyalg}.

\begin{theorem}\label{thm:greedy:ic}
Algorithm \ref{greedyalg} is a truthful mechanism without money and with verification for CAs with unknown $k$-minded bidders.
\end{theorem}
\begin{proof}
Fix $i$ and $\bi$. As in Definition \ref{def:ksetmono}, take two declarations of bidder $i$, $a=(z, {\cal T})$ and $b=(w, {\cal U})$ with $w(U) \geq z(T)$, where $T=A_i(a, \bi)$ and $U=\sigma(T|b)$. (In this proof, $A$ denotes Algorithm \ref{greedyalg}.) Recall that $U \in {\cal U}$ and $U \subseteq T$.

Let $\mathsf{S}_a$ (respectively, $\mathsf{S}_b$) be the set comprised of the sets in declarations of $\bi$ processed by $A(a, \bi)$ (respectively, $A(b, \bi)$) when $z(T)$ (respectively, $w(U)$) is considered. Since $A$ grants $T$ to bidder $i$ in the instance $(a, \bi)$ then it must be the case that $T \cap S = \emptyset$ for all $S \in \mathsf{S}_a$ granted by $A$. Since $w(U) \geq z(T)$, then we have that $\mathsf{S}_b \subseteq \mathsf{S}_a$. Thus, since $U \subseteq T$ then we have that $U \cap S=\emptyset$ for all $S \in \mathsf{S}_b$ granted by the algorithm. Therefore, the only reason for which $U$ might not be granted to $i$ is that $A$ had already granted a set in ${\cal U}$ to $i$, which implies that $w(A_i(b, \bi)) \geq w(U)$. Then the algorithm is $k$-set monotone and the claim follows from Theorem \ref{thm:char:unknown}.
\end{proof}

\cite[Theorem~2]{BL10} proved a lower bound of $\Omega(m)$ on the approximation ratio of any truthful greedy priority mechanism with money 
for instances with demanded sets of cardinality at most $2$. Nevertheless, we have shown that 
Algorithm \ref{greedyalg} is truthful 
and achieves an approximation ratio of at most $3$ for such instances.
We here investigate the reasons behind this sharp contrast.
%
%

\begin{proposition}\label{prop:nopay4greedy}
There are no payments that augment Algorithm \ref{greedyalg} so to make a truthful mechanism for CAs with $k$-minded bidders, even in the case of known bidders.
\end{proposition}


\def\RandExp{\mathrm{RandExp}}
\def\RandPoly{\mathrm{RandPoly}}

\subsubsection{Randomized Exponential-Time Mechanism}
We describe the exponential-time mechanism, or $\RandExp$ in brief. Let $I$ be an instance of CAs with unknown $k$-minded bidders, and let $I_\ell$, $1 \leq \ell \leq k$, be the subinstance of $I$ that consists of the elementary bids $(i, S_i^\ell, v_i(S_i^\ell))$, $i \in N$, where $S_i^\ell$ denotes the $\ell$-th most valuable set demanded by bidder $i$. Then, $\RandExp$ computes the maximum social welfare $\OPT_\ell$ for each subinstance $I_\ell$ by breaking ties among optimal solutions in a bid-independent way, and outputs the allocation corresponding to $\OPT_\ell$ with probability $1/k$, for each $\ell \in [k]$. 

\begin{theorem}\label{t:k-exponential}
$\RandExp$ is a universally truthful mechanism without money and with verification for CAs with unknown $k$-minded bidders. It achieves an approximation ratio of $k$.
\end{theorem}


\subsubsection{Randomized Polynomial-Time Mechanism}
We describe the polynomial-time mechanism, or $\RandPoly$ in brief. Let $I$ be an instance of CAs with unknown $k$-minded bidders, let $v_{\max}$ be the maximum valuation of some bidder, and let $S_{\max}$ a set with valuation $v_{\max}$. Moreover, let $I_s$ be the subinstance of $I$ that consists of the elementary bids $(i, S, v_i(S))$, $i \in N$, where for each set $S$, $|S| \leq \sqrt{m}$. Then, $\RandPoly$ either only allocates $S_{\max}$ to the corresponding bidder breaking ties in a bid-independent way with probability $1/2$, or with probability $1/2$, outputs the allocation computed by the Algorithm \ref{greedyalg} on the subinstance $I_s$. Next, we show the following fact.

\begin{theorem}\label{t:simple_randomized}
$\RandPoly$ is a universally truthful mechanism without money and with verification for CAs with unknown $k$-minded bidders. It achieves an approximation ratio of $O(\sqrt{m})$.
\end{theorem}
%

\section{Lower bounds for known bidders}

%
We first adapt the proof of \cite[Theorem~3.3]{DG10} and show a lower bound of $2$ on the approximation ratio of any deterministic truthful mechanism. We highlight that this lower bound does not make any assumptions on the computational power of the mechanism, and holds even for exponential-time mechanisms.

\begin{theorem}\label{t:det_lb_2}
There are no deterministic truthful mechanisms with approximation ratio
better than $2$ for CAs with known $2$-minded bidders. This
holds even for simple instances with $n = 2$ bidders, and $m = 2$ goods.
\end{theorem}

We note here, that our Algorithm \ref{greedyalg} is a truthful $2$-approximate mechanism on instances used in the proof of Theorem~\ref{t:det_lb_2}. This theorem indicates that our assumption that the bidders do not overbid on their winning sets is 
less powerful than the use of payments, when we do not take computational issues into consideration. Furthermore, it shows that, unlike the case of single-minded bidders, already with double-minded bidders the class of algorithms that can be implemented with money is a strict superset of the class of $2$-monotone algorithms.

%
Next, we apply Yao's principle 
and show that no randomized mechanism that is universally truthful can achieve an approximation ratio better than $5/4$.

\begin{theorem}\label{t:universally_truthful}
There are no randomized mechanisms that are universally truthful and have approximation ratio better than $5/4$ for CAs with known $2$-minded bidders. This holds even for simple instances with $n = 2$ bidders, and $m = 2$ goods.
\end{theorem}

\begin{proof}
We present a probability distribution over instances of CAs with $n = 2$ known $2$-minded bidders, and $m = 2$ goods, for which the best deterministic truthful mechanism has expected approximation ratio greater than $5/4-\delta$, for any $\delta > 0$.
We consider two instances $I$ and $I'$ where $\universe = \{ a, b\}$, the first bidder is interested in ${\cal S}_1 = \{ \{ a, b \}, \{ b \} \}$, and the second bidder is interested in ${\cal S}_2 = \{ \{ a \} \}$. In both, the valuation of bidder $2$ is $v_2(\{ a \}) = 1$. The valuation of bidder $1$ is $v_1(\{ a, b \}) = 2$ and $v_1(\{ b \}) = 0$ in $I$, and $v'_1(\{ a, b \}) = 2$ and $v'_1(\{ b \}) = 2 - \delta$ in $I'$. Each instance occurs with probability $1/2$, and the expected maximum social welfare is $(5-\delta)/2$.
Let $A$, applied to instance $I$, allocate $\{a, b\}$ to bidder $1$ and $\emptyset$ to bidder $2$. Then, by Theorem \ref{thm:char:known}, since $A$ is a deterministic truthful mechanism, when applied to instance $I'$, it must allocate $\{a, b\}$ to bidder $1$ and $\emptyset$ to bidder $2$. 
Therefore, the expected social welfare of $A$ is $2$, and its expected approximation ratio is $(5 - \delta) / 4 > 5/4 - \delta$.
If $A$, applied to instance $I$, does not allocate $\{ a, b\}$ to bidder $1$, its expected social welfare is at most $(4-\delta)/2$, and its expected approximation ratio is $(5 - \delta) / (4 - \delta) > 5/4 - \delta$, a contradiction.
\end{proof}

%
Finally, we show a weaker lower bound of $1.09$ on the approximation ratio achievable by the larger class of randomized mechanisms that are truthful in expectation.

\begin{theorem}\label{t:truthful_in_expectation}
There are no randomized mechanisms that are truthful in expectation and have approximation ratio less or equal than $1.09$ for CAs with known $2$-minded bidders. This holds even for simple instances with $n = 2$ bidders, and $m = 2$ goods.
\end{theorem}

\subsection{Priority Mechanisms}

Our lower bounds in this section apply to truthful mechanisms that use algorithms that operate according to the priority framework introduced in \cite{BNR03}.

We now briefly introduce this framework. The input of a priority mechanism is a finite subset $I$ of the class $\mathcal{I}$ of all permissible input items. For CAs, we consider two classes of input items. The first class consists of elementary bids, i.e., each item 
is a triple $(i, S, v_i(S))$, where $i$ is the bidder, $S$ is one of $i$'s demanded sets, and $v_i(S)$ is $i$'s valuation for $S$. The second class of input items consists of bidders, i.e., each item 
is a pair $(i, v_i)$, where $i$ is the bidder and $v_i$ is $i$'s valuation function.
%
The output of a priority mechanism consists of a decision for each input item processed. If elementary bids are the input items, the output consists of an accept or reject decision for each bid. If a bid $(i, S, v_i(S))$ is accepted, $S$ is allocated to $i$, and the algorithm obtains a welfare of $v_i(S)$. If bidders $(i, v_i)$ are the input items, the output consists of a (possibly empty) set $S$ allocated to bidder $i$, and the algorithm collects a welfare of $v_i(S)$.

A (possibly adaptive) priority mechanism $A$ receives as input a finite set of items $I \subseteq \mathcal{I}$, and proceeds in rounds, processing a single item in each round. While there are unprocessed items in $I$, $A$ selects a total order $\mathcal{T}$ on $\mathcal{I}$ without looking at the set of unprocessed items. It is important that $\mathcal{T}$ can be any total order on $\mathcal{I}$, and that for adaptive priority algorithms, the order may be different in each round. In each round, $A$ receives the first (according to $\mathcal{T}$) unprocessed item $x \in I$ and makes an irrevocable decision for it (e.g., if $x$ is an elementary bid, $A$ accepts or rejects it, if $x$ is a bidder, $A$ decides about the set allocated to her). Then, $x$ is removed from $I$. 


We first show that any truthful priority mechanism $A$ that processes elementary bids has an approximation ratio of at least $d$ for CAs with known bidders. 
The proof of the next result adapts the proof of \cite[Theorem~3]{BL10}. 


\begin{theorem}\label{t:bids_lb}
Let $A$ be a truthful priority mechanism with verification and no money for CAs with known $k$-minded bidders. If $A$ processes elementary bids then the approximation ratio of $A$ is greater than $(1-\delta) d$, for any $\delta > 0$.
\end{theorem}

We note that with a minor change in the proof, the lower bound of Theorem~\ref{t:bids_lb} applies to the special case of $2$-minded bidders. Thus, taking into account instances with $d = m$, we obtain the following result.

\begin{corollary}\label{cor:bids_lb}
Let $A$ be a truthful priority mechanism with verification and no money for CAs with known $2$-minded bidders. If $A$ processes elementary bids then the approximation ratio of $A$ is greater than $(1-\delta) m$, for any $\delta > 0$.
\end{corollary}

We next show that any truthful priority mechanism $A$ that processes bidders has an approximation ratio of at least $m/2$ for CAs with known $k$-minded bidders. We note that such priority mechanism $A$ is potentially more powerful than a priority mechanism that processes elementary bids, since when $A$ decides about the set allocated to each bidder $i$, it has full information about $i$'s valuation function. The proof of the following result adapts 
the proof of \cite[Theorem~4]{BL10} to our setting.

\begin{theorem}\label{t:bidders_lb}
Let $A$ be a truthful priority mechanism with verification and no money for CAs with known $2$-minded bidders. If $A$ processes bidders then the approximation ratio of $A$ is greater than $(1-\delta) m/2$, for any $\delta > 0$.
\end{theorem}

\subsection{Discussion}\label{sec:discussion}
A step that seems necessary for an approximation ratio of $O(\sqrt{m})$ for CAs is that the algorithm somehow compares the social welfare and chooses the best of the following two extreme solutions: a solution that only consists of the most valuable set demanded by some bidder, and a solution consisting of many small sets with a large total valuation. Otherwise, the algorithm cannot achieve an approximation ratio of $o(m)$ even for the simple case where bidder $1$ is double-minded for $\universe = \{ a_1, \ldots, a_m \}$ with valuation $x \in \{ 1+\eps, m^2 \}$ and for the good $a_1$ with valuation $1$, and each bidder $i$, $2 \leq i \leq m$, is single-minded for the good $a_i$ with valuation $1$. In fact, this is one of the restrictions of priority algorithms exploited in the proofs of the lower bounds of $\Omega(m)$ above.

On the other hand, comparing the social welfare of these two extreme solutions seems also sufficient for an $O(\sqrt{m})$-approximation, in the sense that taking the best of (i) the most valuable set demanded by some bidder, and (ii) the solution of Algorithm~\ref{greedyalg}, if we only allocate sets of cardinality at most $\sqrt{m}$, gives an $O(\sqrt{m})$-approximation (see Theorem~\ref{t:simple_randomized} for the analysis, see also e.g., \cite[Section~6]{BKV05} for another example of an $O(\sqrt{m})$-approximation algorithm based on a similar comparison).

For CAs without money, it seems virtually impossible to let a deterministic mechanism truthfully implement a comparison between the social welfare of those extreme solutions. This is because the only way for a deterministic mechanism to make sure that the bidder with the maximum valuation does not lie about it is to allocate her most valuable set to her, so that verification applies to this particular bid (see also how Algorithm~\ref{f:MPU-algo-modified} learns about $v_{\max}$). But this leads to an approximation ratio of $\Omega(m)$. In fact, this seems the main obstacle towards a deterministic truthful $O(\sqrt{m})$-approximate mechanism for CAs with $k$-minded bidders. So, we conjecture that there is a strong lower bound of $\Omega(m)$ on the approximation ratio of deterministic truthful mechanisms.


\bibliographystyle{plain}
\bibliography{cawithver}

\appendix

\section{Postponed proofs}
\subsection{Proof of Lemma \ref{le:ths}}
\begin{proof}
Let us first show that the existence of the thresholds for an algorithm $A$ implies that $A$ is $k$-monotone. Fix $i$, $\bi$ and $t_i$; let $a$ be a declaration in $D_i$ such that $a(S_i^j)>\Theta_i^j(\bi)$ for some $j \in [k]$ and $a(S_i^\ell)<\Theta_i^\ell(\bi)$ for $\ell < j$. Then $A_i(a,\bi)=S_i^j$. Now take $b \in D_i$ such that $b(S_i^j)\geq a(S_i^j)$. No matter what the declaration $b(S_i^h)$ is, for $h \neq j$, by definition of thresholds we have that $A_i(b,\bi)=S_i^l$, $l \leq j$.

For the opposite direction, fix $i$, $\bi$ and $t_i$; let $a$ be a declaration in $D_i$ such that $A_i(a,\bi)=S_i^j$. We prove by induction on $j$ that there exists a threshold $\Theta_i^j(\bi)$ as in the statement. For the base case consider $j=1$. It is straightforward to see that the threshold $\Theta_i^1(\bi)$ exists \cite{LOS}. Now assume that there are thresholds $\Theta_i^1(\bi), \ldots, \Theta_i^{j-1}(\bi)$ and take $b \in D_i$ to be such that $b(S_i^j)\geq a(S_i^j)$ and $b(S_i^h)<\Theta_i^{h}(\bi)$ for $h<j$. By definition of $\Theta_i^1(\bi), \ldots, \Theta_i^{j-1}(\bi)$ and that of $k$-monotonicity, $\Set(A_i(b,\bi)|t_i)=S_i^j$. We can then conclude that there exists a threshold $\Theta_i^j(\bi)$ as well.
\end{proof}

\subsection{Proof of Lemma \ref{lemma-MPU-approx1}}
\begin{proof}
We first prove the first bound. Because the algorithm sells only sets above their thresholds, $v_i(S_i) \ge \sum_{e \in S_i} p_e^i$.
Hence,
\[
v(S) \ge \sum_{i=1}^n \sum_{e \in S_i} p_e^i = \sum_{i=1}^n \sum_{e \in S_i}  p_0 r^{\ell_e^i}
= p_0 \sum_{e \in U} \sum_{k=0}^{\ell_e^*-1} r^k = p_0 \sum_{e \in U} \frac{r^{\ell_e^*}-1}{r-1}
\enspace ,
\]
which gives the bound.

For the second bound of the lemma, consider an optimal feasible allocation $T = (T_1,\ldots,T_n)$. The set choice in line \ref{l:MPU:setchoice} in the algorithm of \cite{KV12} is done by asking the demand oracle.\footnote{Given goods' prices, the demand oracle of a given bidder with valuation function $v_i$ outputs the set $S$ maximizing the difference between the set's valuation $v_i(S)$ and the sum of the prices of its goods.} We observe here that the set choice in this line in Algorithm~\ref{f:MPU-algo} is sufficient for showing this claim. Namely, if $v_i(T_i) \ge \sum_{e \in T_i} p_e^i$, then by the choice of the algorithm $v_i(S_i) \geq v_i(T_i)$, and so we have $v_i(S_i) \ge v_i(T_i) - \sum_{e \in T_i} p_e^i$. If on the other hand $v_i(T_i) < \sum_{e \in T_i} p_e^i$, then $v_i(S_i) \ge v_i(T_i) - \sum_{e \in T_i} p_e^i$ also holds.

Because $p^*_{e} \geq p^i_{e}$, for every $i$ and $e$, this implies $v_i(S_i) \geq v_i(T_i) - \sum_{e \in T_i} p^*_{e}$. Summing over all bidders gives
$
   v(S) = \sum_{i=1}^n v_i(S_i) \geq \sum_{i=1}^n v_i(T_i) - \sum_{i=1}^n \sum_{e \in T_i} p^*_{e} \geq v(T) - b \sum_{e \in U} p^*_{e},
$
where the latter equation uses the fact that $T$ is feasible so that each good is given to at most $b$ sets.
\end{proof}

\subsection{Proof of Theorem \ref{thm:UB:supply:mod:rand}}
\begin{proof}
The produced solution $R$ is feasible, because we ask bidders in line \ref{l:MPU-algo-mod-rand:setchoice} only for sets containing available goods.

For proving the approximation ratio we observe that the only difference between the randomized rounding algorithm from \cite{KV12} and Algorithm \ref{f:MPU-algo-modified-rand} is that our algorithm serves first the max bidder and allocates her best set with probability $1$ rather than with probability $q \leq 1$, the remaining bidders are served as in the former algorithm. We restate Lemma 5 and 4 from \cite{KV12} here. The first of these lemmas is used to prove the second and the second implies our approximation guarantee.

\begin{lemma}\label{lemma-5}
Let us consider CAs with demanded sets of size at most $d \ge 1$ and let $q^{-1} = 2 {\rm e} d^{1/b} \log(4bm)$. For any $1 \le i \le n$ and any set $T \subseteq \universe$ with $|T| \leq d$, we have $\Ex{v_i(T \cap \universe _i)} \ge \frac{1}{2} v_i(T)$.
\end{lemma}

\begin{lemma}~\label{lemma-MPU-approx3}
Suppose the probability $q^{-1} = 2 {\rm e} d^{1/b} \log(4bm)$ is used in Algorithm~\ref{f:MPU-algo-modified-rand}.
Then $\Ex{v(S)} \ge \frac{1}{8}  OPT$ and  $\Ex{v(R)} \ge \frac{q}{8} OPT$.
\end{lemma}

The proof of the former lemma in \cite{KV12} remains unchanged in our case. The only change that is needed in the proof from \cite{KV12} of the latter lemma is firstly to observe that the way of handling the max bidder actually helps the approximation ratio. 
Secondly, as shown in the proof of
the second bound of Lemma \ref{lemma-MPU-approx1},
Algorithm~\ref{f:MPU-algo-modified-rand} chooses in line \ref{l:MPU-algo-mod-rand:setchoice} set $S_i := \mbox{argmax } \{v_i(S): S \in {\cal S}_i \mbox{ such that } S \subseteq \universe _i \mbox{ and } v_i(S) \geq \sum_{e \in S} p^i_e\}$, which implies that $v_i(S_i) \ge v_i(T_i) - \sum_{e \in T_i} p_e^i$, and this is what is needed in the proof.

The universal truthfulness of Algorithm~\ref{f:MPU-algo-modified-rand} is proved in the main body of the paper.
\end{proof}

\subsection{Proof of Theorem \ref{thm:UB:supply:mod:rand:any-m}}
\begin{proof} We only give a sketch here. We use Algorithm \ref{f:MPU-algo-modified-rand} as a subroutine and use a standard randomization on the top of this algorithm. Namely, the mechanism flips a fair coin to choose one out of two algorithms. If the coin shows head, then Algorithm \ref{f:MPU-algo-modified-rand} is executed with parameter $q$ set analogously to the case of sets of size at most $d = \lfloor m^{\frac{b}{b+1}}\rfloor$, that is, with $q^{-1} = 2 {\rm e} d^{1/b} \log(4bm)$. If the coin shows tail, then the mechanism only considers sets of full cardinality~$m$. This setting corresponds to an auction where each bidder wants to buy only a single super item (corresponding to $\universe$) which is available in $b$ copies. The $b$ copies of the super item are sold by calling Algorithm \ref{f:MPU-algo-modified-rand} with $m = 1$ (single item) and $d=1$ (sets of cardinality 1). This mechanism can easily be shown to be universally truthful and a simple analysis shows the claimed approximation ratio, see, e.g., \cite{KV12}.
\end{proof}

\subsection{Proof of Proposition \ref{prop:nopay4greedy}}
\begin{proof}
We consider $3$ instances $I$, $I'$, $I''$, with known bidders and $\universe = \{ a, b \}$. In all of them, bidder $1$ is interested in ${\cal S}_1 = \{ \{ a \}, \{ b \} \}$, and bidder $2$ is interested in ${\cal S}_2 = \{ \{ a \} \}$ and has $v_2(\{ a \}) = 1$. Let $\delta > 0$ be some small constant less than $1/2$. In $I$, bidder $1$ has $v_1(\{a\}) = 1-\delta$ and $v_1(\{b\}) = 0$. In $I'$, bidder $1$ has $v'_1(\{a\}) = 1+\delta$ and $v'_1(\{b\}) = 1$. In $I''$, bidder $1$ has $v''_1(\{a\}) = 1-\delta$ and $v''_1(\{b\}) = 1-2\delta$. Therefore, the allocation of Algorithm~\ref{greedyalg} is $\{ b \}$ to bidder $1$ and $\{ a \}$ to bidder $2$ for instance $I$, $\{ a \}$ to bidder $1$ and $\emptyset$ to bidder $2$ for instance $I'$, and $\{ b \}$ to bidder $1$ and $\{ a \}$ to bidder $2$ for instance $I''$. Now we consider the vertices corresponding to $v_1$, $v_1'$, and $v_1''$ of the declaration graph for Algorithm~\ref{greedyalg} and find an edge $(v_1, v_1')$ of weight $-(1-\delta)$, an edge $(v_1', v_1'')$ of weight $\delta$, and an edge $(v_1'', v_1)$ of weight $0$. Therefore, the declaration graph contains a negative cycle $v_1 \rightarrow v_1' \rightarrow v_1'' \rightarrow v_1$, which implies that the valuation-greedy mechanism cannot be truthfully implemented with money.
\end{proof}

\subsection{Proof of Theorem \ref{t:k-exponential}}
\begin{proof}
Since all bidders are single-minded in each subinstance $I_\ell$, the allocation corresponding to the maximum social welfare $\OPT_\ell$ using a fixed tie-breaking rule is, by Theorem \ref{thm:char:unknown}, truthful, for each $\ell \in [k]$. Therefore, $\RandExp$ is universally truthful. As for the approximation ratio, the expected social welfare of $\RandExp$ is $\sum_{\ell=1}^k \OPT_\ell / k$. Since the maximum social welfare of $I$ is at most $\sum_{\ell=1}^k \OPT_\ell$, $\RandExp$ has an approximation ratio of $k$.
\end{proof}

\subsection{Proof of Theorem \ref{t:simple_randomized}}
\begin{proof}
As for the truthfulness of $\RandPoly$, the deterministic allocation of the set of maximum valuation with a fixed tie-breaking rule is, by Theorem \ref{thm:char:unknown}, truthful. Furthermore, the formation of the subinstance $I_s$ and the application of Algorithm \ref{greedyalg} to it can be regarded as a run of the algorithm on $I$ where any elementary bid $(i, S, v_i(S))$ with $|S| > \sqrt{m}$ is considered unfeasible and immediately rejected. By using the arguments used in Theorem~\ref{thm:greedy:ic}, we can prove that this $\sqrt{m}$-cardinality-sensitive variant of Algorithm \ref{greedyalg} to $I$ is truthful. Therefore, $\RandPoly$ is universally truthful.

As for the approximation ratio, we let $\OPT$ denote the optimal social welfare of $I$, $\OPT_s$ denote the optimal social welfare of the subinstance $I_s$, and $\OPT_l$ denote the optimal social welfare of the subinstance $I_l = I \setminus I_s$. We observe that $\OPT \leq \OPT_s + \OPT_l$. Since $I_l$ contains only elementary bids $(i, S, v_i(S))$ with $|S| > \sqrt{m}$ and $v_i(S) \leq v_{\max}$, $\OPT_l \leq \sqrt{m} v_{\max}$. Moreover, by Theorem~\ref{t:greedy-approx}, the allocation computed by the application of valuation-greedy to $I_s$ is a $(\sqrt{m}+1)$-approximation of $\OPT_s$. Hence, the expected social welfare of $\RandPoly$ is at least:
\[
 \frac{1}{2}\left( v_{\max} + \frac{\OPT_s}{\sqrt{m}+1}\right) \geq \frac{\OPT_s + \OPT_l}{2(\sqrt{m}+1)}
\]
Therefore, the approximation ratio of $\RandPoly$ is $2(\sqrt{m}+1)$.
\end{proof}

\subsection{Proof of Theorem~\ref{t:det_lb_2}}

\begin{proof}
For sake of contradiction, let us assume that there is a deterministic truthful mechanism $A$ with an approximation ratio of $2 - \delta$, for some $\delta > 0$. We consider two instances where $\universe = \{ a, b\}$, and both bidders are in ${\cal S}_1 = {\cal S}_2 = \{ \{ a \}, \{ b \} \}$. In the first instance, $v_1(\{ a \})
= 1+\delta$ and $v_1(\{ b \}) = 1$, and $v_2(\{ a \}) = 1+\delta$ and $v_2(\{ b \}) = 1$. Since $A$ is a $(2-\delta)$-approximation algorithm, $A(v_1, v_2)$ must allocate both sets $\{ a \}$ and $\{ b \}$. Without loss of generality, we assume that $A(v_1, v_2)$ allocates $\{ a \}$ to bidder $1$ and $\{ b \}$ to bidder $2$. Moreover, by Theorem \ref{thm:char:known}, if $v'_2(\{ a \}) = 1+\delta$ and $v'_2(\{ b \}) = 0$, then $A(v_1, v'_2)$ must allocate $\{ a \}$ to bidder $1$ and $\{ b \}$ to bidder $2$. Therefore, $A$ has an approximation ratio of at least $(2+\delta) / (1+\delta)$, which is larger than $2-\delta$, for all $\delta > 0$. \end{proof}

\subsection{Proof of Theorem \ref{t:truthful_in_expectation}}
\begin{proof}
For sake of contradiction, we assume that there is a randomized truthful-in-expectation mechanism $A$ with approximation ratio at most $\rho = 1.09$. 

As in the proofs of Theorems~\ref{t:det_lb_2} and \ref{t:universally_truthful}, we consider instances where $\universe = \{ a, b\}$, the first bidder is interested in ${\cal S}_1 = \{ \{ a, b \}, \{ b \} \}$, and the second bidder is interested in ${\cal S}_2 = \{ \{ a \} \}$. We consider two instances $I$ and $I'$. In both of them, the valuation of bidder $2$ is $v_2(\{ a \}) = 1$. The valuation of bidder $1$ is $v_1(\{ a, b \}) = \varphi$, where $\varphi = (1+\sqrt{5})/2$ is the golden ratio, and $v_1(\{ b \}) = 0$ in $I$, and $v'_1(\{ a, b \}) = \varphi$ and $v'_1(\{ b \}) = 1$ in $I'$.

We assume that $A$ can have only two different solutions for instances $I$ and $I'$. More specifically, either $A$ allocates $\{ a, b \}$ to bidder $1$ and $\emptyset$ to bidder $2$, which happens with probability $p$ for instance $I$ and $q$ for instance $I'$, or $A$ allocates $\{ b \}$ to bidder $1$ and $\{ a \}$ to bidder $2$, which happens with probability $1-p$ for instance $I$ and $1-q$ for instance $I'$. Note that this assumption is without loss of generality, since all the other feasible solutions have worst social welfare than the two considered.

Using that the approximation ratio of $A$ is at most $\rho$, we obtain that:
\begin{align}
\frac{\varphi}{p \varphi + 1 - p} \leq \rho & \Rightarrow
p \geq \frac{\varphi - \rho}{\rho(\varphi - 1)} & \label{eq:inst1}\\
\frac{2}{q \varphi + 2 (1-q)} \leq \rho & \Rightarrow
1 - q \geq \frac{2 - \rho \varphi}{\rho(2 - \varphi)} & \label{eq:inst2}
\end{align}
where (\ref{eq:inst1}) follows from the approximation ratio of $A$ for instance $I$, and (\ref{eq:inst2}) follows from the approximation ratio of $A$ for instance $I'$.

Moreover, since $A$ is truthful in expectation, the expected welfare of bidder $1$ from $A$'s allocation for instance $I$, which is $p \varphi$, does not exceed her expected welfare from $A$'s allocation for instance $I'$, which is $q \varphi + (1 - q)$. Otherwise, bidder $1$ could underbid on $\{ b \}$ by declaring $v_1$, and get an expected welfare of $p \varphi$. Therefore, we obtain that:
\begin{align}
p \varphi \leq q \varphi + 1 - q & \Rightarrow
q \geq \frac{p \varphi - 1}{\varphi - 1} & \label{eq:trf}
\end{align}

Combining (\ref{eq:inst1}), (\ref{eq:inst2}), and (\ref{eq:trf}), we conclude that $\rho$ satisfies that:
\[ \frac{2 - \rho \varphi}{\rho (2 - \varphi)} + \frac{\varphi\frac{\varphi - \rho}{\rho(\varphi - 1)} - 1}{\varphi -1} \leq 1 \]
This is a contradiction, because the inequality above does not hold if $\rho \in [1, 1.09]$. Thus, we conclude that any randomized mechanism $A$ for CAs that is truthful-in-expectation, has an approximation ratio worse than $1.09$.
\end{proof}

\subsection{Proof of Theorem \ref{t:bids_lb}}
\begin{proof}
For sake of contradiction, let us assume that for some given $d$, $2 \leq d \leq m$, there is a truthful priority mechanism $A$ for CAs with known $k$-minded bidders that achieves an approximation ratio of $(1-\delta)d$, for some constant $\delta > 0$.

Let $L$ be any subset of $\universe$ of cardinality $d$. As an input to $A$, we consider an instance $I_1$ that for each bidder $i$, contains elementary bids $(i, L, 1+\delta)$ and $(i, S, 1)$, for all $\emptyset \neq S \subset L$. As a priority mechanism, $A$ selects a bid from $I_1$ and considers it first. In the following, we distinguish between the case where the first bid is $(i, L, 1+\delta)$, for some bidder $i$, and the case where the first bid is $(i, S, 1)$, for some bidder $i$ and some set $S$, and show how to arrive at contradiction in both.

\noindent \underline{Case 1.} Let us assume that the first bid is $(i, L, 1+\delta)$, for some bidder $i$. Then, if $A$ accepts $(i, L, 1+\delta)$, it obtains a social welfare of $1+\delta$, while the optimal social welfare is $d$, which contradicts the hypothesis that the approximation ratio of $A$ is $(1-\delta)d$.
If $A$ rejects $(i, L, 1+\delta)$, we consider $A$'s approximation ratio for an instance $I_2 \subseteq I_1$ that includes only the elementary bid $(i, L, 1+\delta)$. Since $A$ cannot distinguish between $I_1$ and $I_2$, it rejects $(i, L, 1+\delta)$ when considering $I_2$, which leads to an unbounded approximation ratio for $I_2$.

\noindent \underline{Case 2.} Let us assume that the first bid is $(i, S, 1)$, for some bidder $i$ and some set $\emptyset \neq S \subset L$. Then, if $A$ accepts $(i, S, 1)$, we consider $A$'s approximation ratio for an instance $I_3 \subseteq I_1$ that includes the elementary bids $(i, S, 1)$ and $(i, L, 1+\delta)$. Since $A$ cannot distinguish between $I_1$ and $I_3$, it again selects bid $(i, S, 1)$ first and accepts it, when it considers $I_3$. But then consider the instance $I_3'$ in which $i$ changes $(i, S, 1)$ into $(i, S, 1/d)$. 
Since $A$ has an approximation ratio of $(1-\delta)d$, it must allocate $L$ to bidder $i$ in $I_3'$. But this contradicts the hypothesis that $A$ is truthful since the two bids of bidder $i$ in $I_3$ and $I_3'$ do not satisfy $k$-monotonicity (cf. Definition \ref{def:kmono}).

If $A$ rejects $(i, S, 1)$, we consider $A$'s approximation ratio for an instance $I_4 \subseteq I_1$ that includes only the elementary bid $(i, S, 1)$. Since $A$ cannot distinguish between $I_1$ and $I_4$, it rejects $(i, S, 1)$ when considering $I_4$, which leads to an unbounded approximation ratio for $I_4$.
\end{proof}

\subsection{Proof of Theorem \ref{t:bidders_lb}}
\begin{proof}
For sake of contradiction, let us assume that there is a truthful priority mechanism $A$ for CAs with known $2$-minded bidders that processes bidders and achieves an approximation ratio of $(1-\delta)m/2$, for some constant $\delta > 0$.

We consider a universe $\universe = \{ a_1, \ldots, a_m \}$ and $m$ bidders. For each $i$, $1 \leq i \leq m$, we let $g_i$ be a single-minded valuation where the demanded set is $\{ a_i \}$ with valuation $1$. More specifically, for each $S \subseteq \universe$, $g_i(S) = 1$, if $a_i \in S$, and $g_i(S) = 0$, otherwise. Moreover, for each $i$, $1 \leq i \leq m$, we let $f_i$ be a double-minded valuation where the demanded set is either $\{ a_i \}$ with valuation $m^2+\delta$, or $\universe$ with valuation $m^2+2\delta$. More specifically, $f_i(\universe) = m^2+2\delta$, and for each $S \subset \universe$, $f_i(S) = m^2+\delta$, if $a_i \in S$, and $f_i(S) = 0$, otherwise. In the following, we only consider \emph{restricted instances} of CAs where every bidder has a valuation of type either $g$ or $f$.

We first show that for any bidder $i$ and for all instances where all bidders $j \neq i$ have single-minded valuations of type $g$ and $i$ has a valuation of type $f$, $A$ allocates $\universe$ to $i$ and $\emptyset$ to any bidder $j \neq i$ (this claim is the equivalent of \cite[Lemma~5]{BL10} in our setting). We let $f_k$, for some $1 \leq k \leq m$, be the valuation of bidder $i$. Since the optimal social welfare is at least $m^2+2\delta$, $A$ allocates either $\universe$ or a set $S \supseteq \{ a_k \}$ to bidder $i$. Otherwise, the social welfare of $A$ would be at most $m - 1$, which contradicts the hypothesis that the approximation ratio of $A$ is $(1-\delta)m/2$. However, the fact that $A$ is truthful, implies, by Theorem \ref{thm:char:known}, that $A$  must assign $\universe$ to $i$ on input $g$ valuations from bidders other than $i$ and $f_k$ valuation from $i$. Indeed, assume that it is not the case and consider a new instance in which bidder $i$ declares a single-minded valuation where the demanded set is $\universe$ with valuation $m^2+2\delta$. Because of the approximation guarantee of $A$, on this new instance, $A$ must grant $\universe$ to $i$. The two declarations would then contradict $k$-monotonicity (cf. Definition \ref{def:kmono}).
Therefore, $A$ allocates $\universe$ to bidder $i$ and $\emptyset$ to any bidder $j \neq i$.

Using this claim, we can prove the following proposition which is identical to \cite[Lemma~6]{BL10}. The proof is by induction on $i$, and is omitted because it is essentially identical to the proof in \cite
{BL10}. In fact, that proof 
uses only standard properties of priority algorithms, the assumption that the approximation ratio of $A$ is $(1-\delta)m/2$, and \cite[Lemma~5]{BL10}, which, in our case, is replaced by the claim above.

\begin{proposition}\label{prop:half_sequence}
Let $A$ be any truthful priority mechanism which for restricted instances with $m$ goods, achieves an approximation ratio of $(1-\delta)m/2$, for some constant $\delta > 0$. Then, there exists a labeling of the bidders and the goods such that the following holds for all $i \in \{ 0, 1, \ldots, m/2 - 1 \}$. Let instance $I_i = \{ (j, g_j) : 1 \leq j \leq i \}$. Then, for any restricted instance $I \supseteq I_i$, $A$ considers all the bidders in $I_i$ before all other bidders in $I$, and allocates $\emptyset$ to each bidder in $I_i$.
\end{proposition}

Using Proposition~\ref{prop:half_sequence}, we can now complete the proof of the lemma. Let $I' = \{ (j, g_j) : 1 \leq j \leq m/2 -1\}$ be the instance $I_i$ defined in Proposition~\ref{prop:half_sequence} for $i = m/2-1$, and let $I = I' \union \{ (j, g_{m}) : m/2 \leq j \leq m\}$. We note that the optimal social welfare for $I$ is $m/2$, and that $I$ is a restricted instance such that $I' \subseteq I$, as required by Proposition~\ref{prop:half_sequence}. Therefore, mechanism $A$ considers bidders $1, \ldots, m/2 - 1$ first and allocates $\emptyset$ to each of them. Since bidders $m/2, \ldots, m$ are all single-minded for good $a_m$, the social welfare of $A$ for instance $I$ is at most $1$, which contradicts the hypothesis that the approximation ratio of $A$ is $(1-\delta)m/2$.
\end{proof}

\end{document}